%% file: main.tex
\documentclass[twocolumn,pra,superscriptaddress]{revtex4-2}

\usepackage[utf8]{inputenc}
\usepackage{float}
\usepackage{placeins}
\usepackage{qcircuit}
\usepackage{braket}
\usepackage{amsmath}
\usepackage[normalem]{ulem}
\usepackage{todonotes}
\usepackage{xcolor}
\usepackage{multirow}
\usepackage{braket}
\usepackage{soul}

\usepackage{graphicx}
\usepackage{dcolumn}

\usepackage{bm}
\usepackage{bbold}

\usepackage{todonotes}

\usepackage{dsfont}

\usepackage{hyperref}
\usepackage{amsthm}
\usepackage{verbatim} 
\usepackage{circuitikz}
\usepackage{tabularx}
\usepackage{color,soul}
\usepackage{esint}
\usepackage{physics}
\usepackage{amssymb}
\usepackage{subfigure}

\usepackage{tabto}
\usepackage{listings}
\usepackage[english]{babel}
\usepackage{thmtools,thm-restate}

\usepackage{mathtools} 
\graphicspath{{Figures/}}

\usepackage{diagbox} 

\newcommand{\Mod}[1]{\ (\mathrm{mod}\ #1)}
\newcommand{\Bin}[1]{ {\textrm{BIN}(#1)}}

\makeatletter
\renewcommand*\env@matrix[1][*\c@MaxMatrixCols c]{%
  \hskip -\arraycolsep
  \let\@ifnextchar\new@ifnextchar
  \array{#1}}
\makeatother

\begin{document}

\preprint{APS/123-QED}
\title{Tridiagonal matrix decomposition for Hamiltonian simulation on a quantum computer}

\author{Boris Arseniev} \affiliation{Skolkovo Institute of Science and
Technology, Moscow, Russian Federation}
\author{Dmitry Guskov}\affiliation{Skolkovo Institute of Science and
Technology, Moscow, Russian Federation}
\author{Richik Sengupta}\affiliation{Skolkovo Institute of Science and
Technology, Moscow, Russian Federation}
\author{Jacob Biamonte}\thanks{Former address}\affiliation{Skolkovo Institute of Science and Technology, Moscow, Russian Federation}
\author{Igor Zacharov}\affiliation{Skolkovo Institute of Science and
Technology, Moscow, Russian Federation}

\begin{abstract}

The construction of quantum circuits to simulate Hamiltonian evolution is central to many quantum algorithms.
State-of-the-art circuits are based on oracles whose implementation is often omitted, and the complexity of the algorithm is estimated by counting oracle queries.
However, in practical applications, an oracle implementation contributes a large constant factor to the overall complexity of the algorithm.
The key finding of this work is the efficient procedure for representation of a tridiagonal matrix in the Pauli basis, which allows one to construct a Hamiltonian evolution circuit without the use of oracles.
The procedure represents a general tridiagonal matrix $2^n \times 2^n$ by systematically determining all Pauli strings present in the decomposition, dividing them into commuting subsets. The efficiency is in the number of commuting subsets $O(n)$. 
The method is demonstrated using the one-dimensional wave equation, verifying numerically that the gate complexity as a function of the number of qubits is lower than the oracle based approach for $n < 15$ and requires half the number of qubits. 
This method is applicable to other Hamiltonians based on the tridiagonal matrices.

\end{abstract}

\maketitle

\section{\label{sec:level1}Introduction}

Simulation of quantum Hamiltonians is one of the first directions that can potentially demonstrate quantum advantage \cite{feynman2018simulating}. Over the years this field has evolved and great progress has been made since the first description \cite{lloyd1996universal}. The latest developed techniques, such as the quantum walk algorithm \cite{Berry}, simulation by qubitization \cite{Low2019hamiltonian}, and others \cite{childs2003exponential, berry2015simulating, berry2016corrected, low2017optimal}, are often described in terms of calls to an oracle, but the construction of the oracle is usually not specified. In a notable exception from the above, the oracle is constructed in Ref.~\cite{suau2021practical} and the constant factor in gate scalability is calculated.

The standard approach to implement a circuit for Hamiltonian simulation on a quantum computer is to decompose the Hamiltonian $H$ into a sum of Pauli strings (tensor product of Pauli matrices) and approximate the operator $e^{-\imath Ht}$ with product formulas \cite{ Trotter, Childs_Trotter}. With a naive approach, the number of Pauli strings to consider is $4^n$ if the size of $H$ is $2^n \times 2^n$. It was conjectured that modeling $e^{-\imath Ht}$ using this approach is more efficient if one can group the resulting Pauli strings into commuting sets \cite{Kawase, Berg, mukhopadhyay2023synthesizing}. It was shown that the set of all Pauli operators (without identity) of size $4^n -1$ can be divided into $2^n + 1$ different subsets, each consisting of $2^n - 1$ internally commuting elements \cite{unbiased}.

The problem of partitioning a Hamiltonian decomposition in the Pauli basis featuring sets of commuting operators was studied in the framework of simultaneous measurements \cite{Izmaylov}. Typically, this problem may be solved by building a graph with Pauli strings as nodes, connected if they commute, i.e., the clique problem. Further, it can be reduced to the graph-coloring problem which is NP-Complete, but heuristics exist \cite{Cliques}. 

In this work we consider Hamiltonians of a special kind which are constructed using tridiagonal matrices. Tridiagonal matrices come to light in many different areas of mathematical and applied sciences, commonly in the discretization of differential equations \cite{band, LeVeque}, and are used to represent discretized versions of differential operators in quantum computing.

The proposed procedure decomposes tridiagonal matrices into $O(n)$ internally commuting subsets of Pauli strings, each subset having size $O(2^n)$. It also provides the coefficients (weights) for each Pauli string in the decomposition. It automatically leverages the structure of the tridiagonal matrix to remove the majority of the redundant Pauli strings with zero coefficients and provides an upper bound for the number of Pauli strings with non-zero weights. Moreover, it contains a formula for calculation of the weights separate from the symbolic generation of Pauli strings.

We illustrate our method using the Hamiltonian of the one-dimensional wave equation as an example and numerically show the dependence of the number of gates on the number of qubits. We also show that our method for $n<15$ qubits has fewer gates for practical applications than one with the oracle implementation, despite worse theoretical scaling.

This work is organized as follows. In Sec. \ref{sec:definitions} we introduce the notation and useful mathematical constructs. The decomposition algorithm for an arbitrary tridiagonal matrix is described in Section \ref{sec:tridiag_decomp}, followed by specific variants for real and real symmetric tridiagonal matrices. Sec. \ref{sec:H_decomposition} is a special case of a symmetrized matrix $H$ constructed from a real matrix $B$ such that both $B$ and $B^\top$ are on the anti-diagonal. In Sec. \ref{sec:ham_sim} we focus on Hamiltonian simulation, while in Sec. \ref{sec:dif_eq} the method is illustrated with an example of the one-dimensional wave equation. We defer longer proofs to the Appendix.

\section{Notation and definitions} \label{sec:definitions}

The Pauli matrices constitute a basis for the complex vector space of $2 \times 2$ matrices and comprise operators $\mathcal {P}=\{ I,X,Y,Z\}$, where 
\begin{equation*}
\begin{gathered}
    I =
    \begin{pmatrix}
        1 & 0 \\
        0 & 1
    \end{pmatrix}, 
    X = 
    \begin{pmatrix}
        0 & 1 \\
        1 & 0
    \end{pmatrix},\;\
    Y = 
    \begin{pmatrix}
        0 & -\imath \\
        \imath & 0
    \end{pmatrix},\;\;
    Z = 
    \begin{pmatrix}
        1 & 0 \\
        0 & -1
    \end{pmatrix}.
\end{gathered}
\end{equation*}

A tensor product of several Pauli matrices is called a Pauli string. The length of a Pauli string is the number of Pauli operators in the string, and it is exactly $n$ when decomposing a $2^n \times 2^n$ matrix. We denote the set of Pauli strings of length $n$ as $\mathcal{P}_n$.

The Pauli basis decomposition of an arbitrary matrix $B$ is given by
\begin{equation} 
    B = \frac{1}{2^n}\sum_{j=1}^{M} \alpha_j P_j \;, \;\; P_j  \in \mathcal{P}_n,
    \label{decomp_B}
\end{equation}
where $M$ is the number of terms in the decomposition and $\alpha_j \in \mathbb{C}$ in general. Hereinafter we omit the tensor sign when writing Pauli strings (e.g.,\ $X\otimes Y\otimes Z\otimes Y$ is abbreviated as $XYZY$).

Manipulation of Pauli strings is possible with bit arithmetic. For a single bit $x, p \in \mathbb{B} = \{ 0, 1\} $ we use the following notation for powers:
\begin{equation*}
    x^p = x \oplus \overline{p} = x \oplus p \oplus 1,
\end{equation*}
where $\oplus$ is XOR (addition modulo 2). Definitions for strings of bits 
$\mathbf{x} \equiv (x_1,\dots,x_n) \in \mathbb{B}^n$, where $x_j \in \mathbb{B}, j=1,\dots,n$, are summarized in Table \ref{tab:definitions_bit_strings}.

\begin{table}[h]
\renewcommand{\arraystretch}{1.7} 
    \centering
    \begin{tabular}{c|l}
    Notation & Definition \\
    \hline
	$ \overline{\mathbf{x}} $ & Negation, $\overline{\mathbf{x}} = (\overline{x}_1,\dots,\overline{x}_n)$ \\
    \hline
    $\mathbf{x}^\mathbf{y}$ & Exponentiation, $\mathbf{x}^\mathbf{y}=(x_1^{y_1},\dots,x_n^{y_n})$ \\
    \hline
    $\mathbf{x} \cdot \mathbf{y}$ & Inner product, $\mathbf{x} \cdot\mathbf{y} = \sum_{l=1}^n x_ly_l$ \\
    \hline
    $\delta(\mathbf{x},\mathbf{y})$ & Kronecker delta, $\delta(\mathbf{x},\mathbf{y}) = 
    \prod_{j=1}^n\delta_{x_j,y_j}$ \\
    \hline
	
    $ Z^\mathbf{y} $ & $Z^\mathbf{y} \equiv \bigotimes_{l=1}^n Z^{y_l} = Z^{y_1}\otimes\cdots\otimes Z^{y_n}$ 
    \end{tabular}
    \caption{Notational convention. Here $\mathbf{x} = (x_1,\dots,x_n) \in \mathbb{B}^n$, $\mathbf{y} = (y_1,\dots,y_n) \in\mathbb{B}^n$, and $X$ and $Z$ are Pauli matrices.}
    \label{tab:definitions_bit_strings}
\end{table}

For $\mathbf{x} \in \mathbb{B}^n$ we define the vector $\ket{\mathbf{x}} \in (\mathbb{C}^2)^{\otimes n}$ as
\begin{equation}
	\ket{\mathbf{x}} = \bigotimes_{l=1}^n \ket{x_l} = \ket{x_1,\dots x_n}.
\end{equation}
Further, we define the function that converts non-negative integers to binary
\begin{equation}
 \textrm{BIN}:\mathbb{N}\cup\{0\} \rightarrow \mathbb{B}^{n}.
\end{equation}
Note that $\textrm{BIN}$, as well as the bit string $\mathbf{x}$ are encodings of an integer, where the leftmost bit encodes the lower register. For example, the string $1101$ is the encoding of $1\cdot1+ 1\cdot2+ 0\cdot4 + 1\cdot8 = 11$.

To express Pauli strings with $X$ and $Z$ matrices we use bit strings $\mathbf{z}, \mathbf{x}, \mathbf{p} \in \mathbb{B}^n$ and definitions in Table \ref{tab:definitions_bit_strings}:
\begin{align}
	Z^\mathbf{z}\ket{\mathbf{p}} &= \bigotimes_{l=1}^n Z^{z_l}\ket{p_l} = (-1)^{\mathbf{z} \cdot \mathbf{p}}\ket{\mathbf{p}}, \label{eq:action_Z_on_ket} \\
	X^\mathbf{x}\ket{\mathbf{p}} &= \bigotimes_{l=1}^n X^{x_l}\ket{p_l} = \ket{\overline{\mathbf{p}}^\mathbf{x}}. \label{eq:action_X_on_ket}
\end{align}

An arbitrary Pauli string can be defined as the image of the extended Pauli string operator (Walsh function) \\
$\hat{W}:\mathbb{B}^n\times \mathbb{B}^n\rightarrow \mathcal{P}_n$ as follows:
\begin{equation}
\hat{W}(\mathbf{x},\mathbf{z}) = \imath^{\mathbf{x}\cdot\mathbf{z}} X^{\mathbf{x}}Z^{\mathbf{z}} = \bigotimes_{j=1}^{n}\imath^{x_jz_j}X^{x_j}Z^{z_j},
\end{equation}
with the ordinary matrix product between $X^\mathbf{x}$ and $Z^\mathbf{z}$. It can be seen that the Walsh function is bijective. Thus, each Pauli string can be encoded with a unique pair $(\mathbf{x},\mathbf{z})$ and \eqref{decomp_B} can be rewritten as
\begin{equation}\label{eq: decomposition}
    B = \frac{1}{2^n}\sum_{\mathbf{x},\mathbf{z}\in\mathbb{B}^n}\beta_{\mathbf{x},\mathbf{z}}\hat{W}(\mathbf{x},\mathbf{z}),
\end{equation}
where $\beta_{\mathbf{x},\mathbf{z}} \in \mathbb{C}$. There is one-to-one correspondence between $\alpha_j$ from \eqref{decomp_B} and $\beta_{\mathbf{x},\mathbf{z}}$ in \eqref{eq: decomposition}.

By $\{P_1,\;P_2\}^{\otimes{n}}$ we denote the $n$th Cartesian product of the set $\{P_1,\;P_2\}$ with elements interpreted as Pauli strings. For example, $\{P_1,\;P_2\}^{\otimes{2}} \equiv \{P_1,\;P_2\} \otimes \{P_1,\;P_2\} \equiv \{P_1 \otimes P_1,\;P_1  \otimes P_2,\;P_2 \otimes P_1,\;P_2 \otimes P_2\}$. As before, the product $P_k \otimes P_j$ is abbreviated $P_k P_j$.

\section{Tridiagonal matrix decompositions} \label{sec:tridiag_decomp}

We consider an arbitrary tridiagonal matrix $B \in \mathbb{C}^{N \times N}$, where $N = 2^n$ of the following form: 
\begin{equation}
    B = 
    \begin{pmatrix}
    c_1 & a_1 &  0 &  \dots & 0 &  0 & 0 \\
    b_1 & c_2 &  a_2 & \dots & 0 & 0 & 0   \\
    0 &    b_2 & c_3 & \dots & 0 & 0 & 0   \\
    \vdots & \vdots & \vdots & \ddots & \vdots & \vdots & \vdots \\
    0 & 0 & 0 & \dots & c_{N-2} & a_{N-2} & 0  \\
    0 & 0 & 0 & \dots & b_{N-2} & c_{N-1} & a_{N-1} \\
    0 & 0 & 0 & \dots & 0 & b_{N-1} & c_{N} &  \\
    \end{pmatrix}.
    \label{B_tridiag}
\end{equation}

Proposition \ref{theor: sets} provides the maximal possible set of Pauli strings with nonzero coefficients in the decomposition of an arbitrary tridiagonal matrix with complex entries.  We limit our consideration to tridiagonal matrices with only real entries (proposition \ref{theor: real_tridiag_mat}) and further we  consider only tridiagonal symmetric real matrices ($a_i=b_i$) in corollary \ref{real_sym_decomposition} and provide the maximal possible set of Pauli strings with nonzero coefficients in each case, as well as provide a possible partitioning of these strings into sets of internally commuting operators.

\subsection{Pauli strings present in the decomposition}

We formulate the following proposition regarding the decomposition of an arbitrary tridiagonal matrix $B$ with complex entries shown in \eqref{B_tridiag} into Pauli strings:

\begin{restatable}[Decomposition of an arbitrary tridiagonal matrix]{proposition}{decompositionSets}\label{theor: sets} 
An arbitrary tridiagonal matrix $B \in \mathbb{C}^{N \times N}$, where $N \equiv 2^n$, 
 can have Pauli strings in its decomposition with nonzero coefficients only from the union of the following disjoint sets with total cardinality of $(n+1)2^n:$
\addtolength{\tabcolsep}{-1.2pt} 
\begin{center}
\begin{tabular}{cccccccccccc}
\renewcommand{\arraystretch}{2}
\textit{0}. & $\{I, Z\}$ & $\otimes$ & $\{I, Z\}$ & $\otimes$ & $\cdots$ & $\otimes$ & $\{I, Z\}$ & $\otimes$ & $\{I, Z\}$ & $\otimes$ & $\{I, Z\}$ \\ [1ex]
\textit{1}. & $\{I, Z\}$ & $\otimes$ & $\{I, Z\}$ & $\otimes$ & $\cdots$ & $\otimes$ & $\{I, Z\}$ & $\otimes$ & $\{I, Z\}$ & $\otimes$ & $\{X, Y\}$ \\ [1ex]
\textit{2}. & $\{I, Z\}$ & $\otimes$ & $\{I, Z\}$ & $\otimes$ & $\cdots$ & $\otimes$ & $\{I, Z\}$ & $\otimes$ & $\{X, Y\}$ & $\otimes$ & $\{X, Y\}$ \\ [1ex]
\textit{3}. & $\{I, Z\}$ & $\otimes$ & $\{I, Z\}$ & $\otimes$ & $\cdots$ & $\otimes$ & $\{X, Y\}$ & $\otimes$ & $\{X, Y\}$ & $\otimes$ & $\{X, Y\}$ \\ [1ex]
$\vdots$ &  &  & &  &  &  &  &  &  &  & \\ [1ex]
\textit{n-1}. & $\{I, Z\}$ & $\otimes$ & $\{X, Y\}$ & $\otimes$ & $\cdots$ & $\otimes$ & $\{X, Y\}$ & $\otimes$ & $\{X, Y\}$ & $\otimes$ & $\{X, Y\}$ \\ [1ex]
\textit{n}. & $\{X, Y\}$ & $\otimes$ & $\{X, Y\}$ & $\otimes$ & $\cdots$ & $\otimes$ & $\{X, Y\}$ & $\otimes$ & $\{X, Y\}$ & $\otimes$ & $\{X, Y\}$ \\ [-1ex]
\multicolumn{1}{c}{} & \multicolumn{11}{c}{\upbracefill} \\ [-1ex]
\multicolumn{1}{c}{} & \multicolumn{11}{c}{$\scriptstyle n$}
\end{tabular} 
\end{center}
\addtolength{\tabcolsep}{1.2pt}
\end{restatable}

The first step in the procedure for decomposition of a tridiagonal matrix in the Pauli basis \eqref{decomp_B} consists of symbolic generation of Pauli strings starting from all diagonal $\{I,Z\}$ operators to all antidiagonal operators $\{X,Y\}$  by replacing one diagonal operator on the right with an antidiagonal operator at each step $m$. The decomposition weights can be calculated later, based on the selected Pauli strings, see Sec. \ref{subsec:dec_weight}. Note that the cardinality of the union in proposition \ref{theor: sets} is based on the structure of an arbitrary tridiagonal matrix; weight calculation may result in fewer Pauli strings in the decomposition.

We denote the sets in proposition \ref{theor: sets} as
\begin{equation}
    S_{m,\pm}^n = \{I,\;Z\}^{\otimes({n-m})} \otimes \{X,Y\}^{\otimes m}, \quad m=0, \dots, n.
    \label{eq:general_decomposition}
\end{equation}
When $m=n$ ($m=0$) the first (second) tensor product is omitted. Each step $m$ generates $2^n$ Pauli strings, and we divide these strings into two sets named $S_{m,+}$ and $S_{m,-}$, where $+$ indicates that the number of $Y$ operators in every Pauli string in the set is even and $-$ indicates that this number is odd. When $m=0$, there are no $Y$ operators, so we will denote it as $S_{0}$. Note that each $m$ corresponds to a row labeled by $m$ in proposition \ref{theor: sets}.

The bit string notation from Sec. \ref{sec:definitions} provides a concise description of sets $S_{m,\pm}^n$ where the correspondence to $m$ is given by bit string $\mathbf{x}$ as follows:
\begin{align}
\label{eq: subsets_x_z}
\begin{split}
 S_{m,+}^n =\{\hat{W}(\mathbf{x},\mathbf{z}): \mathbf{x} \cdot\mathbf{z} = 0\: (\mathrm{mod}~2),~\mathbf{x} =  V_m^n,~\mathbf{z} \in \mathbb{B}^n\}, \\
  S_{m,-}^n =\{\hat{W}(\mathbf{x},\mathbf{z}): \mathbf{x} \cdot\mathbf{z} = 1\: (\mathrm{mod}~2),~\mathbf{x} =  V_m^n,~\mathbf{z} \in \mathbb{B}^n\},
\end{split}
\end{align}
where $V_m^n$ is a selector with $m$ bits set to one and the following $n - m$ bits set to zero, like: $V_m^n=(\underbrace{1, \dots, 1}_{m}, \underbrace{0, \dots ,0}_{n - m})$, formally:
\begin{equation}
V_m^n = (v_1^m,\dots,v_n^m), \quad
v_j^m = \begin{cases}
        1, \; m\geq j\\
        0, \; m < j
        \end{cases}.
\label{eq: V_m^n}
\end{equation}
The bit string $\mathbf{x}$ corresponds to $m$ such that the first $m$ positions of the bit string set to one, and the remaining ($n-m$) positions are all zeros. To generate each subset $S_{m,\pm}^n$, $\mathbf{z}$ traverses all numbers $0, \dots, n$ and the generated Pauli strings are sorted according to the outcome of $\mathbf{x} \cdot \mathbf{z}$.

Each of the subsets $S_{m,\pm}^n$ contains commuting Pauli strings due to the following proposition: 

\begin{restatable}[Commutativity criterion]{proposition}{lemmacommutationA}
\label{lem: lemma_commutation}
Let $P=\hat{W}(\mathbf{x},\mathbf{z})$ and $Q=\hat{W}(\mathbf{a},\mathbf{b})$ be two Pauli strings of length $n$, where $\mathbf{x},\mathbf{z},\mathbf{a},\mathbf{b} \in \mathbb{B}^n$. Then, $P$ and $Q$ commute iff 
\begin{equation} \label{eq: commutation}
    \mathbf{x} \cdot\mathbf{b} = \mathbf{a}\cdot\mathbf{z} \Mod{2}.
\end{equation}
\end{restatable}
\
\begin{restatable}{corollary}{corcommutation}
\label{cor: pauli commutation}
Let $P=\hat{W}(\mathbf{x},\mathbf{z})$ and $Q=\hat{W}(\mathbf{x},\mathbf{b})$ be two Pauli strings of length $n$, where $\mathbf{x},\mathbf{z},\mathbf{b} \in \mathbb{B}^n$. Let the parity of $Y$ operators in both strings be equal, then $P$ and $Q$ commute.
\end{restatable}

Proposition \ref{lem: lemma_commutation} and corollary \ref{cor: pauli commutation} are proven in Appendix Section \ref{app: sec: theorem_general}. We reach a conclusion that $S_{m,+}^n$ and $S_{m,-}^n$ each are internally commuting subsets. Therefore, the general matrix decomposition will have $2n+1$ internally commuting subsets.

Analogous to the general case, for a real tridiagonal matrix $B$ the following proposition will hold: 

\begin{restatable}[Real tridiagonal matrix]{proposition}{realtridB}
\label{theor: real_tridiag_mat}
        A real tridiagonal matrix $B \in \mathbb{R}^{N \times N}$, where $N = 2^n$,  can have Pauli strings in its decomposition with non-zero coefficients only from the union of $2n+1$ disjoint internally commuting sets $S_{m,\pm}$, $2n$ of which have a cardinality of  $2^{n-1}$ each and one of which is given by $S_{0,\pm}$ and has a cardinality of $2^n.$ The cardinality of the union is $(n+1)2^n.$
\end{restatable}

For the case of a real symmetric tridiagonal matrix $B$ the symmetry will enable additional cancellations in the decomposition, such that only $S_{m,+}$ may be present. This is reflected in the following corollary:

\begin{restatable}[Real symmetric tridiagonal matrix]{corollary}{realsymdecomposition}\label{real_sym_decomposition}
       A real symmetric tridiagonal matrix $B \in \mathbb{R}^{N \times N}$, where $N = 2^n$,  can have Pauli strings in its decomposition with non-zero coefficients only from the union of $n+1$ disjoint internally commuting sets, $n$ of which are given by $S_{m,+}$ and have a cardinality of  $2^{n-1}$ each and one of which is given by $S_{0}$ and has a cardinality of $2^n.$ The cardinality of the union is $(n+2)2^{n-1}.$
\end{restatable}

For convenience we omit the superscript $n$ in $S_{m,+}^n$ and omit the $+$, when the length $n$ and parity is clear from context.

\subsection{Decomposition weights} \label{subsec:dec_weight}

In order to calculate the coefficients (weights) of the proposed decomposition it is easier to separate the calculation of the weights for the ``diagonal"  and ``off-diagonal (anti-diagonal)" subsets. The matrix $B$ can be written as a sum of the diagonal $D$ and off-diagonal $F$ matrices: $B = D + F$. The diagonal matrix $D$ after decomposition consists of Pauli strings of length $n$ in the subset $S_0$ containing only $Z$ and $I$ operators. The coefficients $\beta_{\mathbf{x},\mathbf{z}}$ of corresponding $\hat{W}(\mathbf{x},\mathbf{z})$ within this first internally commuting set can be calculated as
\begin{equation}
    \beta_{\mathbf{0},\mathbf{z}}  = \sum_{p=0}^{2^n-1}(-1)^{\mathbf{z}\cdot\Bin{p}}c_p,
    \label{eq:alpha_diag_B}
\end{equation}
where $c_p$ is a diagonal element of the matrix $B$ as in \eqref{B_tridiag} and $\mathbf{0}$ is $(0,\dots,0)$, $\mathbf{z} \in \mathbb{B}^n$.

For the off-diagonal matrix $F$, the weights $\beta_{\mathbf{x},\mathbf{z}}$ of the decomposition [see \eqref{eq: decomposition}] may be calculated for each $\hat{W}(\mathbf{x},\mathbf{z})$ in a subset $S_{m,\pm}$ given in \eqref{eq:general_decomposition} as follows:
\begin{equation}
   \begin{split}
   \beta_{\mathbf{x},\mathbf{z}} = \sum_{p=0}^{2^n-2}\imath^{\mathbf{x}\cdot\mathbf{z}}(-1)^{\mathbf{z}\cdot\Bin{p}}\delta(\Bin{p+1},\overline{\Bin{p}}^\mathbf{x})  a_p \\
    + \sum_{p=1}^{2^n-1} \imath^{\mathbf{x}\cdot\mathbf{z}}(-1)^{\mathbf{z}\cdot\Bin{p}}\delta(\Bin{p-1},\overline{\Bin{p}}^\mathbf{x}) b_{p-1}
    \end{split}
    \label{eq:alpha_off_B}
\end{equation}
where $(\mathbf{x},\mathbf{z})$ corresponds to \eqref{eq: subsets_x_z}, and $a_p$ and $b_{p}$ are the off-diagonal elements of matrix $B$ as in \eqref{B_tridiag}. 
The expression $\overline{\Bin{p}}^\mathbf{x}$ is calculated as the bit-wise XOR of bit string $\mathbf{x}$ with the inverted bit string $\mathbf{p}$.

Note that \eqref{eq:alpha_off_B} holds only for the decomposition of off-diagonal elements, while \eqref{eq:alpha_diag_B} holds only for the diagonal elements of matrix $B$. These formulas are derived in Appendix \ref{app: sec: dec_formulae}.

\subsection{Visualization of commuting subsets} \label{subsec:visualization}

 The elements of matrix $B$ can be used to calculate the coefficients of Pauli strings in its decomposition using \eqref{eq:alpha_off_B}. In Fig. \ref{contr_plot} we show the correspondence between elements of $B$ and subsets $S_m$ which contain the corresponding Pauli strings. The diagonal consists of Pauli strings containing only $Z$ and $I$ operators and corresponds to $S_0$. For the off-diagonal elements we are left with $n$ subsets $S_{m,\pm}$ specified in \eqref{eq:general_decomposition}. To show the correspondence, consider the off-diagonal component $F$  of a real symmetric matrix $B$ for $n = 4$. For a symmetric matrix (see corollary \ref{real_sym_decomposition}) we have the following subsets with size $2^{n-1}=2^3$: 
\begin{align*}
S_{1,+} = \{&I,Z\} ^{\otimes 3} \otimes \{X\}, \\
S_{2,+} = \{&I,Z\} ^{\otimes 2} \otimes \{XX,YY\}, \\
S_{3,+} = \{&I,Z\} \otimes \{XXX,XYY,YXY,YYX\}, \\
S_{4,+} = \{&XXXX, XXYY, XYXY, XYYX, \nonumber \\
&YXYX, YXXY, YYXX, YYYY \}.
\end{align*}
Figure \ref{contr_plot} illustrates how these subsets correspond to the elements of a $2^4 \times 2^4$ matrix $B$.  Each $S_{m, +},\; m=1, \dots,4$ is given a color and the structure is apparent. The length of the anti-diagonal segment is equal to $2^m$, where $m$ is the number of the $\{ X,Y \}$ Pauli operators on the right in the $S_{m}$ expression.

\begin{figure}[H]
\centering

\includegraphics[angle=-90,scale=0.3]{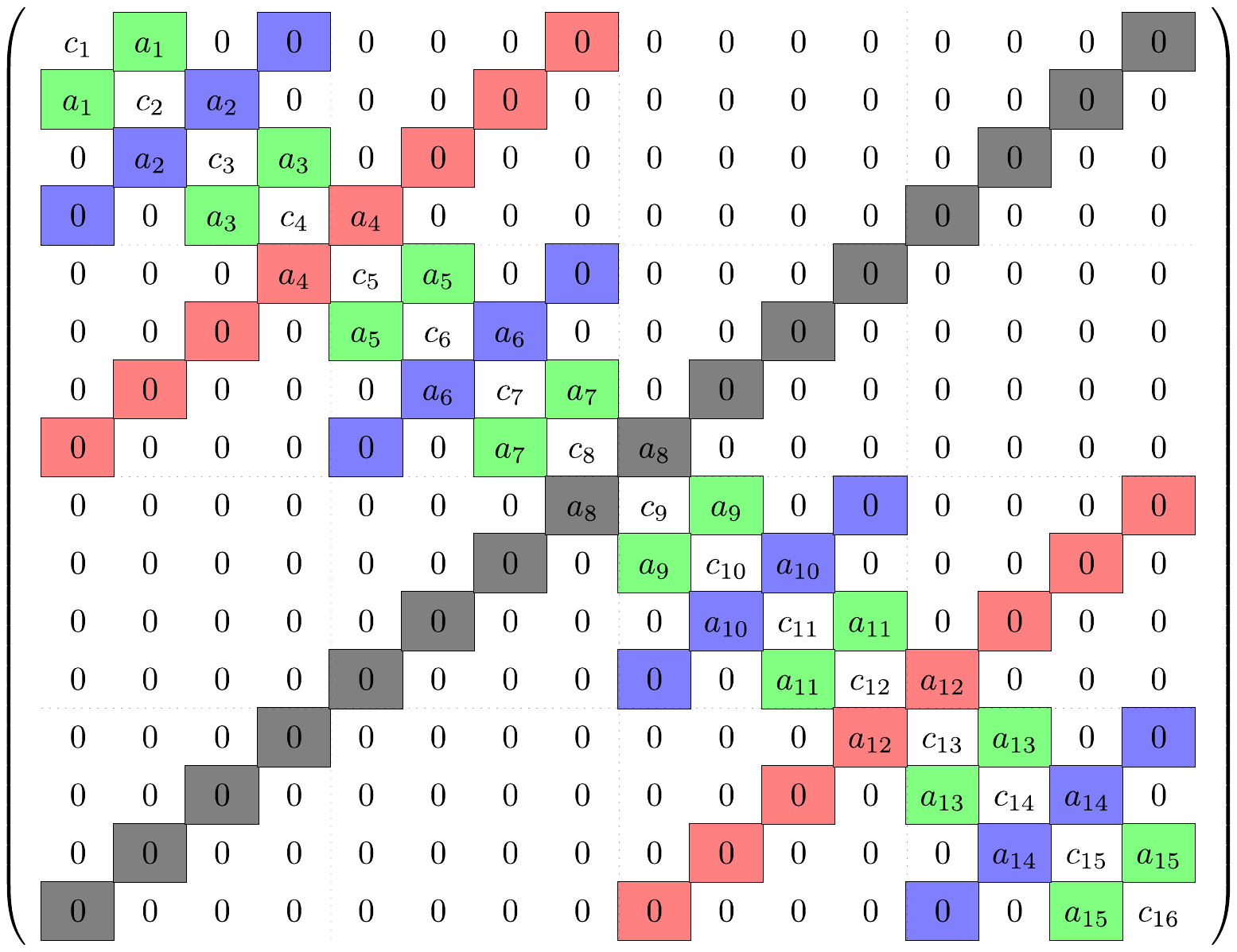}
\caption{
Illustrating the contribution of commuting subsets $S_{m,+}$ to decomposition of the matrix $B$. Same colors contribute to same commuting subset.
Colors: $S_{1}$ -- green, $S_{2}$ -- blue, $S_{3}$ -- red, $S_{4}$ -- black.}
\label{contr_plot}
\end{figure}

The said structure appears not only for real symmetric matrices. It follows from the associativity of the tensor product, where we take the $n-m$ matrices from the set $\{I, Z\}$ on the left with $m$ matrices from the set $\{X, Y\}$ on the right. For each Pauli string $P$ from the subset $S_{m,\pm}$ we have
\begin{equation}
    P = P_D \otimes P_A,
    \label{Pauli_as_DandA}
\end{equation}
where $P_D$ is a diagonal matrix of size $2^{n-m} \times 2^{n-m}$ and $P_A$ is an anti-diagonal matrix of size $2^{m} \times 2^{m}$.

\section{Decomposition for a symmetrized matrix} \label{sec:H_decomposition}
A Hermitian matrix $H$ can be constructed from any matrix $B$ as:
\begin{equation}
    H =
    \begin{pmatrix}
    0 & B\\
    B^\dagger & 0\\
    \end{pmatrix}.
    \label{hamil}
\end{equation}
As before, consider the case where the matrix $B$ is tridiagonal of size $2^n \times 2^n$, therefore the size of $H$ is $2^{n+1} \times 2^{n+1}$. Here we assume $B$ to be real and therefore $H$ is symmetric. The following statement holds:

\begin{restatable}{corollary}{lemmaHermitianA}
\label{lemma_Hermitian}
The number of terms in the decomposition of a symmetric matrix $H$ \eqref{hamil} is equal to the number of terms in the decomposition of the real matrix $B$ and is bounded above by $(n+1)2^n$.
The Pauli strings in the decomposition of $H$ can be partitioned into the following subsets of commuting strings:
\begin{align*} \label{eq:H_decomposition}
    S_0 &= \{X\} \bigotimes_1^{n} \{I,Z\}, &(m=0),\\
    S_m &= \{\widehat{X,Y}\} \bigotimes_1^{n-m} \{I,Z\} \bigotimes_1^{m} \{X,Y\}, \;\; &
    m=1, \dots, n-1,\\
    S_n &= \{\widehat{X,Y}\} \bigotimes_1^{m} \{X,Y\}, &
    (m=n),
\end{align*}
where the expression $\widehat{\{ X, Y\}}$ selects $X$ or $Y$ such that the number of $Y$ operators in each Pauli string is even.
\end{restatable}

The weights of the decomposition of $H$ may be calculated using formulae \eqref{eq:alpha_diag_B} and \eqref{eq:alpha_off_B}. For any square matrix $B$ with its Pauli decomposition given in \eqref{eq: decomposition} the following equality holds:
\begin{equation}
    H = \frac{1}{2^n} \sum_{\mathbf{x},\mathbf{z}\in\mathbb{B}^n} \left( \Re{\beta_{\mathbf{x},\mathbf{z}}} X - \imath \Im{\beta_{\mathbf{x},\mathbf{z}}} XZ \right) \otimes \hat{W}(\mathbf{x},\mathbf{z}),
    \label{eq:alpha_H}
\end{equation}
for details see Appendix \ref{app: sec: theorem_real}.

\begin{figure}[H]
\centering
\includegraphics[angle=-90,scale=0.3]{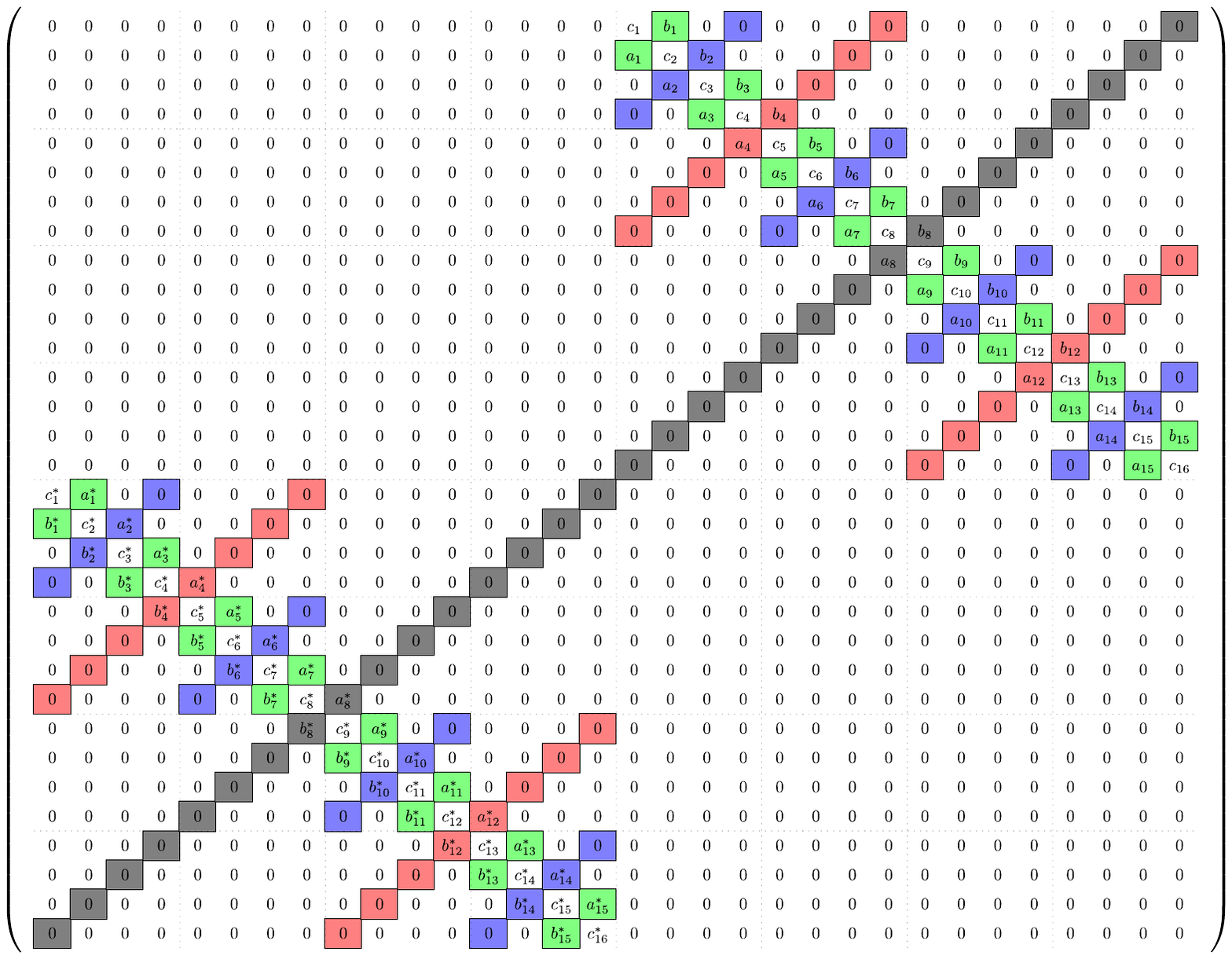}
\caption{Illustrating the contribution of commuting subsets $S_{m}$ to decomposition of the matrix $H$ of size $2^5 \times 2^5$ with $B$ and $B^{\dagger}$ of size $2^4 \times 2^4$. Same colors contribute to same commuting subset.
Color code: $S_1$ -- green, $S_2$ -- blue, $S_3$ -- red, $S_4$ -- black. $S_0$ corresponds to the diagonal of each $B$ matrix.}
\label{fig:H_sets}
\end{figure}

Consider a symmetric matrix $H$ \eqref{hamil} consisting of $B$ and $B^\top$ with $n=4$ as an example. According to corollary \ref{lemma_Hermitian}, Pauli strings in decomposition of $H$ can be arranged into the following sets:
\begin{center}
    \addtolength{\tabcolsep}{-1.2pt}  
    \begin{tabular}{cccccccccc}
    $S_0 =$ & $X$ & $\otimes$ & $\{I, Z\}$ & $\otimes$ & $\{I, Z\}$ & $\otimes$ & $\{I, Z\}$ & $\otimes$ & $\{I, Z\}$, \\ [1ex]
    $S_1 =$ & $\widehat{\{ X,Y\}}$ & $\otimes$ & $\{I, Z\}$ & $\otimes$ & $\{I, Z\}$ & $\otimes$ & $\{I, Z\}$ & $\otimes$ & $\{X, Y\}$, \\ [1ex]
    $S_2 =$ & $\widehat{\{ X,Y\}}$ & $\otimes$ & $\{I, Z\}$ & $\otimes$ & $\{I, Z\}$ & $\otimes$ & $\{X, Y\}$ & $\otimes$ & $\{X, Y\}$, \\ [1ex]
    $S_3 =$ & $\widehat{\{ X,Y\}}$ & $\otimes$ & $\{I, Z\}$ & $\otimes$ & $\{X, Y\}$ & $\otimes$ & $\{X, Y\}$ & $\otimes$ & $\{X, Y\}$, \\ [1ex]
    $S_4 =$ & $\widehat{\{ X,Y\}}$ & $\otimes$ & $\{X, Y\}$ & $\otimes$ & $\{X, Y\}$ & $\otimes$ & $\{X, Y\}$ & $\otimes$ & $\{X, Y\}$.
    \end{tabular}
    \addtolength{\tabcolsep}{1.2pt}
\end{center}
These are similar to subsets which arise in the decomposition of a real symmetric matrix $B$, but now the size of each subset is $2^n$. This follows from the fact that the terms with an odd number of $Y$ operators, which were zero for the real symmetric case, now appear in the decomposition because the $Y$ operator can be selected as the leading term to make the number of $Y$ operators even (see corollary \ref{real_sym_decomposition}). Figure \ref{fig:H_sets} shows which elements correspond to which set for the matrix $H$ of size $2^5 \times 2^5$.

\section{Circuit for Hamiltonian Simulation} \label{sec:ham_sim}

We have shown the two steps in matrix decomposition, namely the generation of Pauli strings based on matrix structure and the calculation of decomposition weights. 
Importantly, Pauli strings are assembled into commuting sets. The information about commuting sets may serve to reduce the circuit complexity of Hamiltonian simulation \cite{Berg} and to accelerate simulation of quantum dynamics on a classical computer \cite{Kawase}. In this section we use this approach to construct a circuit for Hamiltonian simulation.

The parameters that should be taken into account when implementing the evolution of the Hamiltonian \cite{Low2019hamiltonian} include the number of system qubits $n$, evolution time $t$, target error $\epsilon$, and how information on the Hamiltonian $H$ is accessed by the quantum computer. Our circuit does not require additional qubits, and it does not use oracles to access information. As we organized a large number of Pauli strings into an exponentially smaller number of commuting subsets we could expect improvement in the accuracy according to the Trotterization formula. 

The task of Hamiltonian simulation is to implement the operator $e^{-\imath Ht}$ on a quantum device. For an operator $\hat{H}$ represented by a tridiagonal or symmetrized (as discussed in Section \ref{sec:H_decomposition}) matrix $H$ the quantum circuit is constructed as follows.

First we generate internally commuting sets of Pauli strings $\{ S_{m,\pm}^n\}$ needed for the decomposition of $H$. These strings can be simultaneously diagonalized \cite{Kawase, Berg}. As an example, Table \ref{tab:diagonal_ops} contains the diagonalization operators $D_k$ for the real symmetrical tridiagonal matrix of size $2^4$ discussed in Section \ref{subsec:visualization}. Note that $D_0$ is an identity since the $S_0$ set is already diagonal. These procedures, i.e.~the generation of the commuting sets of Pauli strings and optional simultaneous diagonalization are determined by the tridiagonal matrix structure. These steps do not need to be repeated when the matrix changes values, making this approach suitable for use in variational algorithms \cite{cerezo2021variational, peruzzo2014variational}.
\begin{table}[h]
    \renewcommand{\arraystretch}{1.3} 
    \centering
    \begin{tabular}{c|l}
    $D_k$ & Simultaneous diagonalization operators for subset  $S_k$ \\
    \hline
    $D_1$ & $HHHI$ $CX(1,4)$ $CX(2,4)$ $CX(3,4)$ $HHHH$\\
    \hline
    $D_2$ & $HHHI$ $CX(1,3)$ $CX(1,4)$ $CX(2,3)$ $CX(2,4)$ $CX(3,4)$ \\
          & $CZ(1,4)$ $CZ(2,4)$ $CZ(3,4)$ $HHHH$ \\
    \hline
    $D_3$ & $HHHI$ $CX(1,3)$ $CX(1,4)$ $CX(2,4)$ $CX(3,4)$\\
          & $CZ(1,4)$ $CZ(2,4)$ $CZ(3,4)$ $HHHH$ \\
    \hline
    $D_4$ & $HHHI$ $CX(1,2)$ $CX(1,4)$ $CX(2,3)$ \\
          & $CZ(1,3)$ $CZ(3,4)$ $HHHH$
    \end{tabular}
    \caption{Diagonalization operators $D_k$ for a real symmetric tridiagonal matrix of size $2^4$ consisting of the Hadamard operators $H$, the Controlled-NOT $CX(c,t)$, and Controlled-Z $CZ(c,t)$ operators, where $c$ and $t$ are the control and the target qubits, respectively. Qubits are labeled from $1$ to $4$.}
    \label{tab:diagonal_ops}
\end{table}

When a quantum state evolves in time under the action of $e^{-\imath Ht}$ and the matrix $H$ is represented in Pauli basis as in formula \eqref{decomp_B} [or, equivalently, as in \eqref{eq: decomposition}] the Lie-Trotter product formula \cite{Trotter} and its higher order variants \cite{Childs_Trotter} should be implemented taking the internally commuting sets into account. Given the time evolution governed by the Hermitian $H$ and the number of Trotter repetitions $r$, using our decomposition we have:
\begin{align}
e^{- \imath tH} &=\exp \left( \frac{- \imath t}{2^n} \sum_{\mathbf{x},\mathbf{z}\in\mathbb{B}^n} \beta_{\mathbf{x},\mathbf{z}} \hat{W}(\mathbf{x},\mathbf{z}) \right)  \nonumber \\
& = \left[ \prod_{k=0}^{2n} \exp \left( \frac{- \imath t}{2^n r} \tilde{S}_k \right) \right]^{r} + \epsilon,
\label{sum_product}
\end{align}
where $\tilde{S}_k = \sum_{\mathbf{z}} \beta_{\mathbf{x}_k,\mathbf{z}} \hat{W}(\mathbf{x}_k,\mathbf{z})$ are renumbered Pauli strings from subsets $S_{m,\pm}$  defined in \eqref{eq: subsets_x_z}, which contain up to $2^n$ terms each. The approximation is constructed with a Lie-Trotter product formula for $2n+1$ internally commuting sets $S_{m,\pm}$ of Pauli strings, and the accuracy can be estimated as 
\begin{equation}
    \epsilon = \norm{e^{- \imath tH} - \left[ \prod_{k=0}^{2n} \exp \left( \frac{- \imath t}{2^n r} \tilde{S}_k \right) \right]^{r}} = O\left(\frac{t^2}{r}\right),
    \label{error_1st_order}
\end{equation}
where $\norm{\cdot}$ denotes the spectral norm.
The expression \eqref{sum_product} can be diagonalized as
\begin{multline}
    \label{eq:product_groups}
    \prod_{k=0}^{2n} \exp(\frac{- \imath t}{2^n r} \tilde{S}_k) \\
    =\prod_{k=0}^{2n} \exp\left(\frac{-\imath t}{2^n r}\sum_{\mathbf{z}}\beta_{\mathbf{x}_k,\mathbf{z}}\hat{W}(\mathbf{x}_k,\mathbf{z}) \right) \\
    =\prod_{k=0}^{2n} D^\dag_k \exp \left(\frac{-\imath t}{2^n r}\sum_{\mathbf{z}} \beta_{\mathbf{x}_k,\mathbf{z}}\Lambda_{k,\mathbf{z}} \right) D_k ,
\end{multline}
where $D_k$ is the diagonalization operator for each subset $\tilde{S}_{k}$. The number $k$ fixes the string $\mathbf{x}_k$, and the inner summation over $\mathbf{z}$ covers the commuting subset as in \eqref{eq: subsets_x_z}. The operator $\Lambda_{k,\mathbf{z}}$ is a diagonal operator corresponding to $\hat{W}(\mathbf{x}_k,\mathbf{z})$. It is important to note that the coefficients $\beta$ remain the same after diagonalization, or in other words, diagonalization depends only on the structure of the matrix (i.e., depends only on the Pauli subset) and does not depend on the values of the matrix elements. The computation of $D_k$ is done by using Clifford algebras \cite{Kawase, Berg}. The $D_k$ operators consist of the combinations of the single-qubit Hadamard operators and the two-qubit $CX$ and $CZ$ operators; number of gates in $D_k$ scales as $O(n^2)$, where $n$ is number of qubits.

It can be seen that in order to evaluate the propagator $e^{-\imath Ht}$ as in \eqref{sum_product}, \eqref{eq:product_groups} one has to compute the coefficients $\beta_{\mathbf{x},\mathbf{z}}$ and implement $2n+1$ diagonal exponents $\Lambda_{k,\mathbf{z}}$. The coefficients can be computed by using formulae \eqref{eq:alpha_diag_B} and \eqref{eq:alpha_off_B} and in order to implement $2n+1$ diagonal exponents one can use the results from \cite{Welsh} where it is shown that the gate count for a circuit which implements the diagonal exponent can be reduced to $O(N')$ gates with $N' < 2^n$, without ancillary qubits.

Combining all the results together we obtain gate complexity estimated as $O(r (2n+1)(2n^2 + 2^n)) = O(r n 2^n)$ since for one Trotter step for each of $2n+1$ commuting sets one needs to implement $D$, $D^\dagger$ and diagonal exponents. When considering the Trotter formula of an arbitrary order $p$, each Trotter step will contain $O(5^{\lfloor p/2 \rfloor})$ times more gates, but the accuracy $\epsilon$ will be achieved in fewer steps $r$. In work \cite{Childs_Trotter} the Trotterization error considered by leveraging information about commutation and the following scaling is given $\epsilon = O(\frac{\alpha_{\text{comm}} t^{p+1}}{r^{p}})$ with $\alpha_{\text{comm}} = \sum_{j_1, \dots, j_{p+1}}^{M} \norm{[H_{j_{p+1}}, \cdots [H_{j_2}, H_{j_1}] \cdots]}$ and $H = \sum_{j=1}^{M} H_j$, where $H_j$ are anti-Hermitian (which is the case since simulation of $e^{tH}$ is considered in \cite{Childs_Trotter}). It is an interesting question whether it is possible to obtain some tight upper bound for $\alpha_{\text{comm}}$. For now, we will use scaling provided by authors in \cite{berry2007efficient}, i.e. $\epsilon = O\left(\frac{\left(2M5^{\lfloor p/2 \rfloor - 1} \norm{H}t \right)^{p+1}}{r^{p}}\right)$, where $H = \sum_{j=1}^{M} H_j$. Note that in this formula $H$ can be considered as the sum of $2n+1$ matrices formed from commuting sets. Thus $M = 2n + 1$ and resulting number of gates $g$ is given by
\begin{equation}
    g = O\left(t n^2 2^n 5^{p} \norm{H} \left( \frac{t n \norm{H}}{\epsilon}\right)^{1/p} \right).
    \label{resulted_gates}
\end{equation}

This scaling can be made more accurate by taking into account information about the commuting sets. This can be seen for the case where all Pauli strings commute, making the Trotter error equal to zero. It is possible to find some order of Pauli strings that will reduce the error, but since the number of all possible combinations grows exponentially, this is a difficult task \cite{tranter2019ordering}. 

\section{Quantum simulation example: Solving the wave equation} \label{sec:dif_eq}

Tridiagonal matrices arise when discretizing derivatives. For example, the solution of the heat equation $u_t(x,t) = (\kappa (x) u_x(x,t))_x$ may be written as \cite{LeVeque}
\begin{equation*}
    u(x,t) = e^{Bt}u(x,0), 
\end{equation*}
where $B$ is real symmetric tridiagonal matrix as in Section \ref{sec:tridiag_decomp}.

Another example is the wave equation, considered here in more detail. Wave equation in one dimension with amplitude $u(x,t)$ and speed $c(x)$ defined in the interval $x \in [0, 1]$ is given by:
\begin{equation}
u_{tt}(x,t)=(c^2(x)u_x(x,t))_x , \quad u(x,t) \in \mathbb{R}.
\label{eq:1d_wave-q}
\end{equation}
We consider the case of the Dirichlet boundary conditions and set the initial conditions for $u$ and $u_x$ as follows:
\begin{equation}
   \begin{split}
   u(0,t) & = u(1,t) = 0 ,\quad t \in \mathbb{R_+}, \\
   u(x,0)  & = g(x),\quad u_x(x,0)  = 0.
   \end{split}
   \label{eq:wave_initial}
\end{equation}

Following Costa et al.~\cite{Costa} we reduce \eqref{eq:1d_wave-q} to the Schr{\"{o}}dinger equation. Thus, consider the Hamiltonian in the following form ($h$ is space discretization step):
\begin{equation}
    H=\frac{1}{h}     
    \begin{pmatrix}
    0 & B\\
    B^\dagger & 0\\
    \end{pmatrix},
    \label{eq:hamil_wave}
\end{equation}
which leads to the Schr{\"{o}}dinger equation (we use natural units such that $\hbar = 1$) with a two component quantum state $\psi=(\phi_V,\phi_E)^\top$
\begin{equation*}
\frac{d}{dt}
    \begin{pmatrix}
        \phi_V \\ \phi_E 
    \end{pmatrix}
    = \frac{-\imath}{h}
    \begin{pmatrix}
        0 & B \\
        B^\dagger &0  \\
    \end{pmatrix}
    \begin{pmatrix}
        \phi_V \\ \phi_E
    \end{pmatrix}.
\end{equation*}
This recovers the original (discretized) wave equation \eqref{eq:1d_wave-q} if $-BB^\dagger = \mathcal{L}$, where $\mathcal{L}$ is the Laplacian giving an approximation of the second order space derivative. To apply the method proposed in this work, we need the matrix $B$ to be square and have the size $2^n \times 2^n$, so we slightly change the matrix $B$ proposed by Costa et al.~\cite{Costa} by writing the Dirichlet boundary conditions explicitly and we have also incorporated $c(x)$ into this matrix:
\begin{equation*}
    B_{N \times N} = 
    \begin{pmatrix}
     0  & 0     & 0     & 0     & \dots & 0        & 0           &0    \\
     0     & -c_2  & \;\;c_3   & 0     & \dots & 0        & 0           &0    \\
     0     & 0     & -c_3  & \;\;c_4   & \dots & 0        & 0           &0    \\
     \vdots & \vdots & \vdots & \vdots & \ddots & \vdots    & \vdots       & \vdots\\
     0     & 0     & 0     & 0     & \dots & -c_{N-2} & \;\;c_{N-1}     & 0 \\
     0     & 0     & 0     & 0     & \dots & 0        & -c_{N-1}    & c_{N} \\
     0     & 0     & 0     & 0     & \dots & 0        & 0           &0    \\
    \end{pmatrix}.
    \label{B_variable_speed}
\end{equation*}
The wave speed values $c_k, \; k=1, \dots, N$ with $N=2^n$  result from the discretization of the speed profile $c(x)$. The resulting Hamiltonian $H$ has the form described in Section \ref{sec:H_decomposition}. Based on \eqref{eq:alpha_H} the coefficients of the decomposition are given by:
\begin{equation}
\begin{aligned}
    \beta_{\mathbf{0},\mathbf{z}} &= - \sum_{k=1}^{2^n - 2} (-1)^{\mathbf{z} \cdot \Bin{k}}  \cdot c_{k+1},\\
    \beta_{\mathbf{x},\mathbf{z}} &= \sum_{k=1}^{2^n-2} \imath^{\mathbf{x} \cdot \mathbf{z}}(-1)^{\mathbf{z} \cdot \Bin{k}}\cdot\delta(\Bin{k+1},\overline{\Bin{k}}^\mathbf{x})\cdot c_{k+2}, \\
\label{eq:alpha_B_wave}
\end{aligned}
\end{equation}
where $\mathbf{x} = \mathbf{x}(m,n) = V_m^n$ as defined in \eqref{eq: V_m^n}.

Note that with our decomposition the expansion weights contain wave speeds explicitly. Therefore this provides a method to solve partial differential equations with variable coefficients (piece-wise constant over the discretization step $dx$), since decomposition can be done only once and weights recalculated. 

We have implemented the solution for wave equation and use it with a constant speed $c=1$ to determine the number of Trotter steps and the corresponding total number of gates needed to reach the set accuracy $\epsilon<10^{-5}$ (evolution time is $t=1$). This is shown in Figure \ref{fig:number_gates_10-5} as function of the number of qubits $n$ supporting discretization $N=2^n$.
\begin{figure}[H]
\centering
\includegraphics[scale=0.25]{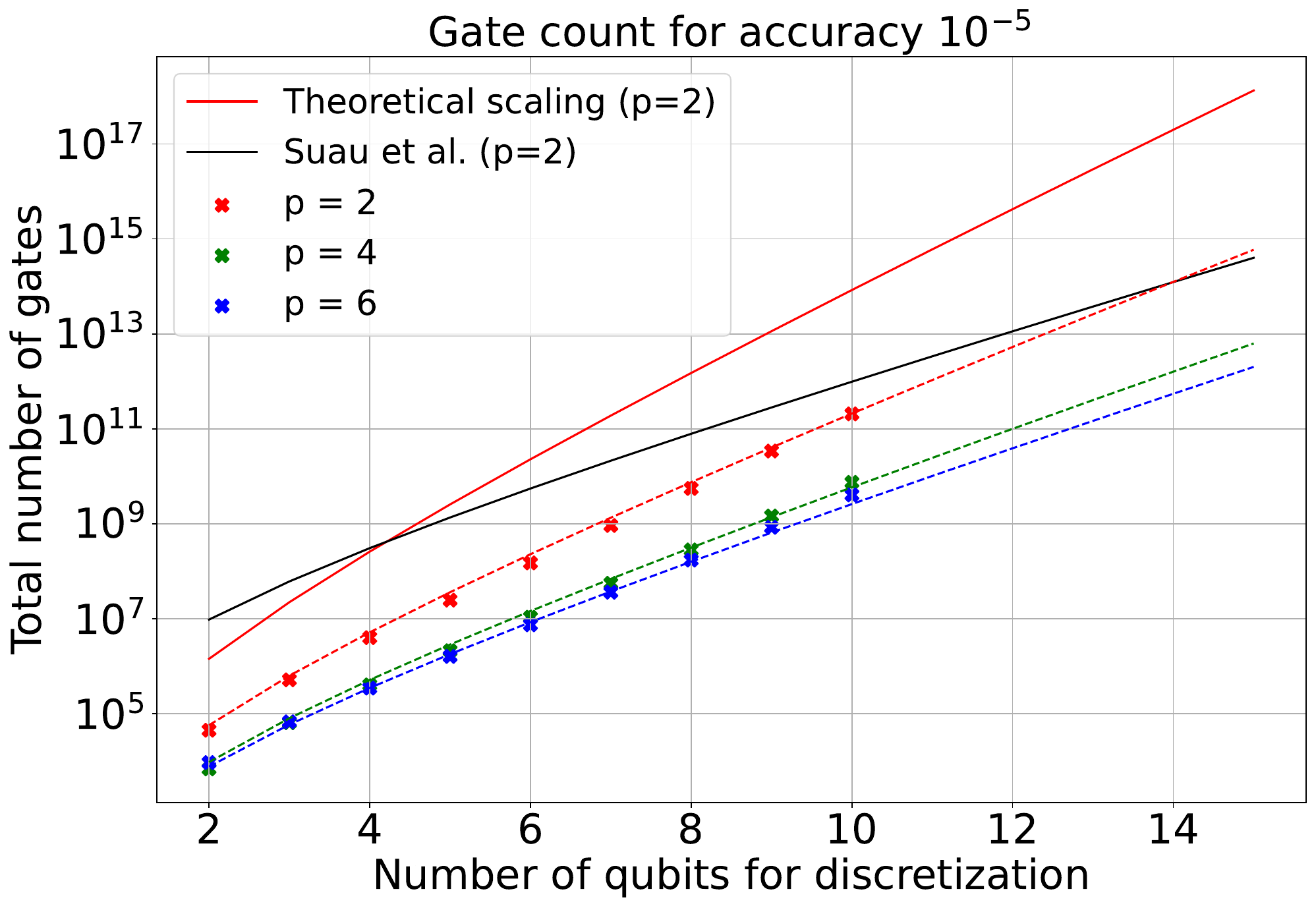}
\caption{Number of gates to approximate $e^{-\imath Ht}$ with accuracy $\epsilon = 10^{-5}$ for $1D$ wave equation Hamiltonian \eqref{eq:hamil_wave}. Points represent number of gates obtained in the simulation.  Red solid line shows the theoretical gate scaling given by \eqref{eq:g1d} for Trotter order $p=2$. The black solid line shows number of gates from reference \cite{suau2021practical}. Dashed lines are fit to the simulation data with \eqref{eq:approx}.} 
\label{fig:number_gates_10-5}
\end{figure}

The gate complexity of our algorithm for solving the wave equation is calculated as 
\begin{equation}
g_{1d} = O\left(5^{p} t n^2 N^2 \left( \frac{t n N}{\epsilon} \right)^{1/p}\right),
\label{eq:g1d}
\end{equation}
using $\norm{H} \leq |\frac{4 c_{\max}}{h}| = O(N)$ and for $p=2$ is shown with a solid line in Figure \ref{fig:number_gates_10-5}. The dashed lines in Figure \ref{fig:number_gates_10-5} represent manual approximation to experimental points with
\begin{equation}
\begin{aligned}
& g_{a} = \gamma N^{\nu} \log(N)^{\mu} = \gamma N^{\nu} n^{\mu}, \;\; \text{with}\\
&\;\;\; \nu = 1.5+1/p, \;\; \mu = 2+1/p, \\
\end{aligned}
\label{eq:approx}
\end{equation}
and the constant factor
\begin{equation*}
\begin{aligned}
&\gamma = 2~5^{p/2 - 1} 10^{5/p}, \;\; & p = 2, 4, \\
&\gamma = ~~5^{p/2 - 1} 10^{5/p}, \;\; & p = 6.
\end{aligned}
\end{equation*}

The actual number of gates $g_a$ scales better than the theoretical one $g_{1d}$ in formula \eqref{resulted_gates}. This is due to the factor $\nu = 1.5+1/p < 2$ in the exponent of $2^n=N$, the number of the discretization points. The constant scaling factor is about $\sim 100$, the implementation is economical. Note that due to the two-component Hamiltonian used in the solution \eqref{eq:hamil_wave} of the Schr{\"{o}}dinger equation the actual number of qubits is $n+1$.

Theoretical gate complexity for the approximation error $\epsilon$ using an oracle is calculated in \cite{suau2021practical} and is shown in \ref{fig:number_gates_10-5} with a solid black line. In this reference authors show oracle implementation of the algorithm presented in \cite{Costa, Berry} with the gate complexity given by $O\left(5^{p} t n^2 N \left( \frac{t N}{\epsilon} \right)^{1/p}\right)$. Here the scaling factor is $\sim O(N^{1+1/p})$, which is better than ours ($\nu=1.5+1/p$ in \eqref{eq:approx}) but with a constant factor of $\sim 300000$. 

Therefore, when comparing our implementation with oracle based for Trotter order $p=2$ we obtain smaller gate count for number of qubits $n<15$.
The oracle based algorithm is also using $\ge 2n$ qubits to implement oracles and thus is less economical than our approach.

\section{Conclusion}
We have presented an effective procedure for decomposition of a $N\times N$ tridiagonal matrix where $N=2^n$ into $2n+1$ subsets of commuting Pauli strings. Each of these subsets has $2^n$ Pauli strings of length $n$ in a general case. Significant for applications is the decomposition of a hermitian matrix consisting of two real tridiagonal matrices of the given type on the anti-diagonal. 
For such matrix there are $n+1$ internally commuting subsets with $2^n$ Pauli strings of length $n$ each as shown in corollary \ref{lemma_Hermitian}. The suggested decomposition procedure considers only non-zero Pauli strings candidates, therefore improving on the brute-force method which examines all $4^{n+1}$ possible Pauli strings. To our knowledge, this work on decomposition of specifically tridiagonal matrices is novel.

This advantage shows up in the calculation of the decomposition coefficients (weights, Section \ref{subsec:dec_weight}), because only the potentially non-zero Pauli strings participate in the evaluation. For the hermitian matrix mentioned above there are $O(N \log{N})$ binary multiplications in the evaluation of one expansion coefficient: we need to compute $O(1)$ products of bit strings with length $O(\log{N})$ for each term (see \eqref{eq:alpha_off_B}).
This compares favorably with the brute force approach using trace and matrix multiplications with complexity $O(N^3)$. The Pauli matrices are sparse, which reduces the complexity of the brute force method to $O(N^2)$ (e.g. \cite{yuster2005fast}). Still, the presented decomposition procedure has exponentially fewer multiplications, which, to our knowledge, is also novel result.

An additional advantage of the presented decomposition is the automatic availability of the commuting subsets. Each of these subsets can be simultaneously diagonalized. This presents an opportunity for complexity reduction when evaluating the Hamiltonian.

For a practical demonstration of the proposed decomposition procedure we constructed a circuit for Hamiltonian simulation using an example of the one-dimensional wave equation (Section \ref{sec:ham_sim}). It is a case where a tridiagonal matrix naturally arises in the discretization of the differential equation, while the Hamiltonian is of the type considered above. There are have been numerous studies of this example in the literature, in particular there is an implementation using an oracle for the Hamiltonian evolution \cite{suau2021practical}. 

Our main result is captured in fig. \ref{fig:number_gates_10-5} where we show that the computational complexity, specifically the number of gates needed to reach the accuracy of $10^{-5}$ scales better than a theoretical estimate. It is also better than the oracle implementation \cite{suau2021practical} for small circuits $n<15$ with an additional advantage that the number of qubits needed to implement the Hamiltonian evolution with the presented method is by a factor of two smaller, than the oracle approach for any size matrix.

We believe that our method can be applied to the $5$-, $7$- and more-diagonal matrices which arise in the discretization of differential in $2D$ and $3D$. Therefore, this will be a natural extension of our study. 

We have formulated our results using Walsh operators which lift to a map on boolean strings. Similar results can be obtained by using Pauli strings directly. However, we believe, using this approach, the propositions on commuting sets and the other relevant results can be expressed algebraically in a concise and simple form.

Finally, we believe that the found commuting subsets are minimal, however a rigorous prove is deferred to future work. 

The Python code for the numerical experiment presented in Figure \ref{fig:number_gates_10-5} is available on \href{https://github.com/barseniev/tridiagonal-matrix-decomposition-quantum-simulation}{GitHub} \cite{github_page}.

\section{Acknowledgements}
\textit{Competing interests.} The authors declare no competing interests. 
\textit{Author contributions.} All authors conceived and developed the theory and design of this study and verified the methods.
The authors acknowledge the use of Skoltech’s Zhores supercomputer \cite{zhores} for obtaining the numerical results presented in this paper.
\bibliography{bibliography}

\pagebreak

\onecolumngrid

\input{appendix}

\end{document}

%% file: appendix.tex
\appendix

\section{Proofs} \label{sec:appendix_proof}

\subsection{Formulae for decomposition coefficients} 
\label{app: sec: dec_formulae}

We call matrix $B$ an upper $l$-diagonal matrix if it has the following form

\begin{equation}
      B = 
\begin{pmatrix}
0 &   \dots & a_0 &  0 & \dots& 0 \\
\vdots & 0 &  \dots & a_1 & \vdots &  \vdots  \\ 
\vdots & \vdots & \ddots & \vdots & \ddots & \vdots \\
\vdots & \vdots & \vdots & \ddots & \vdots & a_{2^n-1-l} \\
\vdots & \vdots &  \vdots & \vdots &\ddots&  \vdots &  \\
0 & \dots &  \dots & \dots &\dots&  0 &  \\
\end{pmatrix}
\end{equation}

and can be written as
\begin{equation}
 \label{eq: l-diagonal}
    B = \sum_{k=0}^{2^n-1-l} a_k \ketbra{\Bin{k}} {\Bin{k+l}} = \sum_{p,q=0}^{2^n-1-l} a_p \delta_{\Bin{p+l},\Bin{q}}\ketbra{\Bin{p}} {\Bin{q}} .
\end{equation}

Similarly, lower $l$-diagonal matrix is introduced as transposed upper $l$-diagonal matrix thus we will not consider this case separately and limit ourselves to upper $l$-diagonal matrix. We also note that Pauli strings present in the decomposition of some matrix $B$ will be the same for $B^\top$ because all Pauli matrices except $Y$ are symmetric and $Y^\top = -Y$, so after transposition Pauli strings in decomposition will be the same but coefficients may change a bit.

\begin{restatable}{proposition}{lemmaCondition}
\label{lem: lem 1}
    Let $B$ be an upper $l$-diagonal matrix. If a Pauli string $P$ enters the Pauli string decomposition of matrix $B\in  \mathbb{C}^{2^n \cross 2^n}$ non-trivially then $\exists p \in \{0, \dots, 2^{n}-1-l\}$:
    \begin{equation}\label{eq: band}
        \Bin{p+l} = \Bin{p}\oplus \mathbf{x} ,
    \end{equation}
    where $\mathbf{x}$ and $\mathbf{z}$ is such that $P = \hat{W}(\mathbf{x},\mathbf{z})$. 
\end{restatable}

\begin{proof}
    Let $B \in \mathbb{C}^{2^n \cross 2^n}$, then decomposition into standard basis may be written as
    \begin{equation}
        B = \sum_{\mathbf{p},\mathbf{q}\in\mathbb{B}^n} b_{\mathbf{p},\mathbf{q}}\ketbra{\mathbf{p}}{\mathbf{q}}.
    \end{equation}
    On the other hand, decomposition into Pauli basis in $\mathbb{C}^{2^n \cross 2^n}$ takes form
\begin{equation}\label{eq: decomposition_app}
    B = \frac{1}{2^n}\sum_{\mathbf{x},\mathbf{z}\in\mathbb{B}^n}\beta_{\mathbf{x},\mathbf{z}}\hat{W}(\mathbf{x},\mathbf{z}),
\end{equation}
where coefficients $\beta_{\mathbf{x},\mathbf{z}} \in \mathbb{C}$. These coefficients can be found by taking the inner product of matrix $B$ and $\hat{W}$ and using formulae \eqref{eq:action_Z_on_ket} and \eqref{eq:action_X_on_ket} from the main text.
\begin{align} \label{eq: decomp_coef}
\begin{split}
    \beta_{\mathbf{x},\mathbf{z}} &= \Tr(B\hat{W}(\mathbf{x},\mathbf{z}))
    =
    \Tr(\sum_{\mathbf{p},\mathbf{q}\in\mathbb{B}^n} b_{\mathbf{p},\mathbf{q}}\ketbra{\mathbf{p}}{\mathbf{q}}\imath^{\mathbf{x}\cdot\mathbf{z}}X^\mathbf{x}Z^\mathbf{z})
    =
    \sum_{\mathbf{s}\in\mathbb{B}^n}\bra{\mathbf{s}} \left( \sum_{\mathbf{p},\mathbf{q}\in\mathbb{B}^n}b_{\mathbf{p},\mathbf{q}}\ketbra{\mathbf{p}}{\mathbf{q}}\imath^{\mathbf{x}\cdot\mathbf{z}}X^\mathbf{x}Z^\mathbf{z} \right) \ket{\mathbf{s}}
    \\
    & = \sum_{\mathbf{s}\in\mathbb{B}^n}\sum_{\mathbf{p},\mathbf{q}\in\mathbb{B}^n}\bra{\mathbf{s}}b_{\mathbf{p},\mathbf{q}}\ketbra{\mathbf{p}}{\mathbf{q}}\imath^{\mathbf{x}\cdot\mathbf{z}}(-1)^{\mathbf{z}\cdot\mathbf{s}}\ket{\mathbf{\overline{s}^\mathbf{x}}} 
    =
    \sum_{\mathbf{s}\in\mathbb{B}^n}\sum_{\mathbf{p},\mathbf{q}\in\mathbb{B}^n}b_{\mathbf{p},\mathbf{q}}\braket{\mathbf{s}}{\mathbf{p}}\braket{\mathbf{q}}{\mathbf{\overline{s}^\mathbf{x}}}\imath^{\mathbf{x}\cdot\mathbf{z}}(-1)^{\mathbf{z}\cdot\mathbf{s}}
    \\
    & = \sum_{\mathbf{s}\in\mathbb{B}^n}\sum_{\mathbf{p},\mathbf{q}\in\mathbb{B}^n}b_{\mathbf{p},\mathbf{q}}\delta_{\mathbf{s},\mathbf{p}}\delta_{\mathbf{q},\mathbf{\overline{s}^\mathbf{x}}}\imath^{\mathbf{x}\cdot\mathbf{z}}(-1)^{\mathbf{z}\cdot\mathbf{s}} 
    =
    \sum_{\mathbf{p},\mathbf{q}\in\mathbb{B}^n} b_{\mathbf{p},\mathbf{q}}\delta_{\mathbf{q},\overline{\mathbf{p}}^\mathbf{x}}\imath^{\mathbf{x} \cdot \mathbf{z}}(-1)^{\mathbf{z} \cdot \mathbf{p}} 
    = 
    \sum_{\mathbf{p}\in\mathbb{B}}\imath^{\mathbf{x}\cdot\mathbf{z}}(-1)^{\mathbf{z}\cdot\mathbf{p}}b_{\mathbf{p},\overline{\mathbf{p}}^\mathbf{x}}.
    \end{split}
\end{align}
Since from \eqref{eq: l-diagonal} we have
\begin{equation}
    B = \sum_{p,q=0}^{2^n-1-l} a_p \delta_{\Bin{p+l},\Bin{q}}\ketbra{\Bin{p}} {\Bin{q}},
\end{equation}
we can obtain $b_{\mathbf{p},\mathbf{q}}$ by equating the coefficients for $\ketbra{\mathbf{p}}{\mathbf{q}}$ and $\ketbra{\Bin{p}} {\Bin{q}}$:
\begin{equation}\label{eq: band_el}
    b_{\mathbf{p},\mathbf{q}}= b_{\Bin{p},\Bin{q}} =  a_p \delta_{\Bin{p+l},\Bin{q}}.
\end{equation}
Substituting this expression into \eqref{eq: decomp_coef}, we obtain final formula for the coefficients

\begin{align}
    \begin{split}
    \beta_{\mathbf{x},\mathbf{z}}  = \sum_{p=0}^{2^n-1-l}\imath^{\mathbf{x}\cdot\mathbf{z}}(-1)^{\mathbf{z}\cdot\Bin{p}}\delta_{\Bin{p+l},\overline{\Bin{p}}^\mathbf{x}}a_p.
            \end{split}
  \end{align}
  
Hence for $\beta_{\mathbf{x},\mathbf{z}} \neq 0$ there should exist a solution $p \in \{0, \dots, 2^{n}-1-l\}$ to the following equation
\begin{equation}
    \Bin{p+l} = \overline{\Bin{p}}^\mathbf{x} = \Bin{p}\oplus\Bin{x}.
\end{equation}
\end{proof}

\subsection{Proof of proposition  \ref{theor: sets} for decomposition of general tridiagonal matrix} 
\label{app: sec: theorem_general}

\lemmacommutationA*
\begin{proof}
\begin{align}
\begin{split}   
    PQ &= \hat{W}(\mathbf{x},\mathbf{z})\hat{W}(\mathbf{a},\mathbf{b}) \\
    &= \bigotimes_{l=1}^{n}  \imath^{x_l\cdot z_l} X^{x_l}Z^{z_l} \imath^{a_l\cdot b_l} X^{a_l}Z^{b_l}  \\
    &= \bigotimes_{l=1}^{n}  \imath^{x_l\cdot z_l} \imath^{a_l\cdot b_l} (-1)^{z_l\cdot a_l} X^{x_l} X^{a_l} Z^{z_l} Z^{b_l}\\
    &= \bigotimes_{l=1}^{n}  \imath^{x_l\cdot z_l} \imath^{a_l\cdot b_l} (-1)^{z_l\cdot a_l} X^{a_l} X^{x_l} Z^{b_l} Z^{z_l}\\
    &= \bigotimes_{l=1}^{n}  \imath^{x_l\cdot z_l} \imath^{a_l\cdot b_l} (-1)^{z_l\cdot a_l} (-1)^{x_l\cdot b_l} X^{a_l} Z^{b_l} X^{x_l} Z^{z_l} \\
    &= (-1)^{\mathbf{z}\cdot\mathbf{a}}(-1)^{\mathbf{x}\cdot\mathbf{b}} \bigotimes_{l=1}^{n}  \imath^{x_l\cdot z_l} \imath^{a_l\cdot b_l} X^{a_l} Z^{b_l} X^{x_l} Z^{z_l} \\
    &= (-1)^{\mathbf{z}\cdot\mathbf{a}}(-1)^{\mathbf{x}\cdot\mathbf{b}} \hat{W}(\mathbf{a},\mathbf{b})\hat{W}(\mathbf{x},\mathbf{z}) \\
    &= (-1)^{\mathbf{z}\cdot\mathbf{a}}(-1)^{\mathbf{x}\cdot\mathbf{b}} QP
\end{split}
\end{align}

Thus we see that two Pauli strings $P$ and $Q$ commute if $\mathbf{z}\cdot\mathbf{a}+\mathbf{x}\cdot\mathbf{b} = 0 \Mod{2}$.

\end{proof}

\corcommutation*
\begin{proof}

We denote the number of $Y$ operators in Pauli string $P$ as $N_Y(P)$. Let $P$ correspond to $\hat{W}(\mathbf{x},\mathbf{z})$ 
\begin{equation}
    \hat{W}(\mathbf{x}, \mathbf{z}) = \bigotimes_l \imath^{x_l z_l}X^{x_l}Z^{z_l}.
\end{equation}
Note that if and only if for some $l,$ $x_l = z_l = 1$ then 
\begin{equation}
    \imath^{x_l z_l}X^{x_l}Z^{z_l} = \imath X Z = Y.
\end{equation}
Therefore, the number of all bit positions where $x_l = z_l = 1$ corresponds to the number of $Y$ in the Pauli string i.e.
\begin{equation} \label{eq: Ny}
    N_Y(P) = \mathbf{x}\cdot\mathbf{z}.
\end{equation}
Since we have $\mathbf{x}\cdot\mathbf{z} = \mathbf{x}\cdot\mathbf{b} \Mod{2}$ from our assumption, then from proposition \ref{lem: lemma_commutation} (here $\mathbf{a}=\mathbf{x}$) a commutation of $P$ and $Q$ follows.
Which means \eqref{eq: commutation} holds and therefore $P$ and $Q$ commute.

\end{proof}

\decompositionSets* 
\begin{proof}

A Tridiagonal matrix decoposition comprises of terms for the diagonal ($l=0$), upper $1$-diagonal and lower $1$-diagonal (that is transposed upper $1$-diagonal). To obtain upper bounds on the number of Pauli strings in the decomposition of $B$ it suffices to check the number of pairs $(\mathbf{x},\mathbf{z})$ that satisfy the necessary conditions from proposition \ref{lem: lem 1}.

\begin{enumerate}
\item For $l=0$ in \eqref{eq: band} we get 
    \begin{gather}
        LHS = \Bin{p+0} = ( p_1,\dots, p_n), \\
        RHS = \overline{\Bin{p}}^\mathbf{x} = (\overline{p}_1^{x_1},\dots,\overline{p}_n^{x_n}).
    \end{gather}
    
    Thus, for a Pauli string to satisfy the necessary conditions $\mathbf{x}=(0,\dots,0)$ and $\mathbf{z}$ can be arbitrary.
    Therefore, the maximum number of Pauli strings in the decomposition of diagonal matrix is bounded by $2^n$.
    
\item For $l=1$ in \eqref{eq: band} using binary summation rules (recall that leftmost bit encodes lowest register) we get
    \begin{equation}
        \Bin{p+1} = ( \overline{p_1}, p_2^{\overline{p_1}},p_3^{\overline{p_1}\overline{p_2}},\dots ).
    \end{equation}
    Since $x^p = x \oplus \overline{p}$, we have
    \begin{align*}
        (\Bin{p+1})_1 &= \overline{p}_1, \\
        (\Bin{p+1})_j &= p_j^{\overline{\prod_{k=1}^{j-1} p_k}} = p_j\oplus\prod_{k=1}^{j-1} p_k, \;\; j>1.
    \end{align*}
    Now we can write \eqref{eq: band} in the following form
    \begin{equation}
        p_j\oplus x_j =p_j\oplus\prod_{k=1}^{j-1} p_k, \;\; j>1,
    \end{equation}
    which means 
    \begin{equation} \label{eq: 1-diag eq}
        x_j = \prod_{k=1}^{j-1} p_k
    \end{equation}
    for $j>1$ and $x_1=1$. It follows that for any $p \in \mathbb{B}^n$ if $p_j=0$ then $x_{j+1}, \dots, x_n$ are equal to $0$.
    Therefore, $x$ can only be one of the following strings:
    \begin{align} \label{eq: 1-diagonal solution}
        \begin{split}
            &(1,0,\dots,0), \\
            &(1,1,\dots,0),\\
            &\vdots \\
            &(1,1,\dots,1).
        \end{split} 
    \end{align}
\end{enumerate}
Therefore, the maximum number of Pauli strings in the decomposition of upper $1$-diagonal matrix is bounded by $n2^n$.

In conclusion, adding the upper bounds for the diagonal, upper and lower $1$-diagonal cases we obtain that the number of terms in the decomposition of $B$ is upper bounded by $(n+1)2^n$ because Pauli strings for upper and lower $1$-diagonal matrices will be the same as was mentioned before. Each of the above strings corresponds to a particular set in the decomposition of $B$ provided in the proposition.
\end{proof}

\subsection{Proof of proposition \ref{theor: real_tridiag_mat} for decomposition of real tridiagonal matrix} \label{app: sec: theorem_real}

\realtridB*
\begin{proof}
The condition that the matrix B is real can be expressed as:
\begin{equation}\label{eq: real_valB}
    B = B^*,
\end{equation}
where $*$ means complex conjugation. From definitions \eqref{eq: decomposition_app} and \eqref{eq: real_valB} it follows:
\begin{align}
    \begin{split}
        LHS &= \frac{1}{2^n}\sum_{\mathbf{x},\mathbf{z}\in\mathbb{B}^n} \beta_{\mathbf{x},\mathbf{z}}\hat{W}(\mathbf{x},\mathbf{z}) \\
        &= \frac{1}{2^n}\sum_{\mathbf{x},\mathbf{z}\in\mathbb{B}^n} \beta_{\mathbf{x},\mathbf{z}}\imath^{\mathbf{x}\cdot\mathbf{z}}X^\mathbf{x}Z^\mathbf{z}, \\
        RHS &= \frac{1}{2^n}\sum_{\mathbf{x},\mathbf{z}\in\mathbb{B}^n} \beta_{\mathbf{x},\mathbf{z}}^*\hat{W}^*(\mathbf{x},\mathbf{z}) \\
        &= \frac{1}{2^n}\sum_{\mathbf{x},\mathbf{z}\in\mathbb{B}^n} \beta_{\mathbf{x},\mathbf{z}}^*(-\imath)^{\mathbf{x}\cdot\mathbf{z}}X^\mathbf{x}Z^\mathbf{z} \\
        &= \frac{1}{2^n}\sum_{\mathbf{x},\mathbf{z}\in\mathbb{B}^n} \beta_{\mathbf{x},\mathbf{z}}^*(-1)^{\mathbf{x}\cdot\mathbf{z}}\imath^{\mathbf{x}\cdot\mathbf{z}}X^\mathbf{x}Z^\mathbf{z}  \\
        &= \frac{1}{2^n}\sum_{\mathbf{x},\mathbf{z}\in\mathbb{B}^n} \beta_{\mathbf{x},\mathbf{z}}^*(-1)^{\mathbf{x}\cdot\mathbf{z}}\hat{W}(\mathbf{x},\mathbf{z}).
    \end{split}
\end{align}

By comparing $LHS$ and $RHS$ one can obtain conditions on coefficients $\beta_{\mathbf{x},\mathbf{z}}$:
\begin{equation} \label{eq: real_condition}
    \beta_{\mathbf{x},\mathbf{z}} = \beta^*_{\mathbf{x},\mathbf{z}}(-1)^{\mathbf{x}\cdot\mathbf{z}}.
\end{equation}

Thus, $\beta_{\mathbf{x},\mathbf{z}}$ has only the real part if $\mathbf{x}\cdot\mathbf{z} = 0 \Mod{2}$ or has only imaginary part if $\mathbf{x}\cdot\mathbf{z} = 1 \Mod{2}$.

From Corollary \ref{cor: pauli commutation}, Pauli strings commute if the parity of the number of $Y$ operators $N_Y = \mathbf{x}\cdot\mathbf{z}$ in each string is the same. Hence, we can partition the set of Pauli strings in terms of $\mathbf{x}$ in proposition \ref{theor: sets} into commuting subsets due to the fact that for every string $\mathbf{x}$, there are two sets of $\mathbf{z}$ with equal cardinality $2^{n-1}$ that result in $N_Y$ being zero or one modulo 2. However, for $\mathbf{x}=(0,\dots,0)$, any $\mathbf{z}$ will not change $N_Y$. Consequently, we will end up with $2n$ subsets, each with cardinality $2^{n-1}$, and one subset with $2^n$ elements.
\end{proof}

\realsymdecomposition*

\begin{proof}
Let's recall the decomposition of $B$:

\begin{equation}
    B = \frac{1}{2^n}\sum_{\mathbf{x},\mathbf{z}\in\mathbb{B}^n} \beta_{\mathbf{x},\mathbf{z}}\hat{W}(\mathbf{x},\mathbf{z}).
\end{equation}

Since $X^\top = X$, $Z^\top=Z$ and $XZ = -ZX$  for an arbitrary Pauli operator $\hat{W}(\mathbf{x},\mathbf{z})$ we can write 
\begin{equation}
    \hat{W}^\top(\mathbf{x},\mathbf{z}) = \imath^{\mathbf{x}\cdot\mathbf{z}}Z^\mathbf{z}X^{\mathbf{x}} = 
    \imath^{\mathbf{x}\cdot\mathbf{z}} (-1)^{\mathbf{x}\cdot\mathbf{z}}X^\mathbf{x}Z^\mathbf{z}=(-1)^{\mathbf{x}\cdot\mathbf{z}}\hat{W}(\mathbf{x},\mathbf{z}).
\end{equation}
Since matrix $B$ is symmetric
\begin{equation}
    B = B^\top = \frac{1}{2^n}\sum_{\mathbf{x},\mathbf{z}\in\mathbb{B}^n} \beta_{\mathbf{x},\mathbf{z}}\hat{W}^\top(\mathbf{x},\mathbf{z}) = 
    \frac{1}{2^n}\sum_{\mathbf{x},\mathbf{z}\in\mathbb{B}^n} \beta_{\mathbf{x},\mathbf{z}}(-1)^{\mathbf{x}\cdot\mathbf{z}}\hat{W}(\mathbf{x},\mathbf{z}).
\end{equation}
 Comparing $RHS$ and $LHS$ we obtain a condition for coefficients $\beta_{\mathbf{x},\mathbf{z}}$:
 \begin{equation}
     \beta_{\mathbf{x},\mathbf{z}} = \beta_{\mathbf{x},\mathbf{z}}(-1)^{\mathbf{x}\cdot\mathbf{z}},
 \end{equation}
 which can be satisfied for non-trivial $\beta_{\mathbf{x},\mathbf{z}}$ only when
 \begin{equation} \label{eq: dag_condition}
     \mathbf{x}\cdot\mathbf{z} = 0 \Mod{2}.
 \end{equation}
 
If $\mathbf{x}=0$, condition \eqref{eq: dag_condition} automatically satisfied, for each of the $n$ remaining possible strings $\mathbf{x}$ only half of the possible $\mathbf{z}$ strings satisfy \eqref{eq: dag_condition} and thus form subsets of size $2^{n-1}$.
 
 Also it follows from \eqref{eq: dag_condition} that the number of $Y$ operators $N_Y$ must be even. Therefore the decomposition of real symmetric tridiagonal matrix consists only of subsets $S_{m, +}$. Thus, we have $n$ subsets $S_{m, +}$ of size $2^{n-1}$ and one subset $S_0$ of size $2^n$.
\end{proof}

\lemmaHermitianA*
\begin{proof}
The Hermitian matrix $H$ is expressed in terms of matrix $B$: 
\begin{equation}
    H =
    \begin{pmatrix}
    0 & B\\
    B^\dagger & 0\\
    \end{pmatrix},
\end{equation}
Let us introduce the following matrices
\begin{align}
    \begin{split}
    E_{12} = \begin{pmatrix}
        0 & 1 \\
        0 & 0
    \end{pmatrix} = \frac{1}{2}(X-XZ), \\
    E_{21} = \begin{pmatrix}
    0 & 0 \\
    1 & 0
\end{pmatrix} = \frac{1}{2}(X+XZ).
    \end{split}
\end{align}
Then $H$ can be constructed as
\begin{equation}
\begin{split}
    H & = E_{12} \otimes B  + E_{21} \otimes B^{\dagger} \\
    & = \frac{1}{2}(X-XZ)\otimes B  + \frac{1}{2}(X+XZ) \otimes B^{\dagger} \\
    & = \frac{1}{2} [ X \otimes \left( B + B^{\dagger} \right) - XZ \otimes \left( B - B^{\dagger}\right)] \\
    & = \frac{1}{2^n} X \otimes \left(\sum_{\mathbf{x},\mathbf{z}\in\mathbb{B}^n} \frac{\beta_{\mathbf{x},\mathbf{z}} + \beta^*_{\mathbf{x},\mathbf{z}}}{2} \hat{W}(\mathbf{x},\mathbf{z}) \right) - 
    \frac{1}{2^n} XZ \otimes \left(\sum_{\mathbf{x},\mathbf{z}\in\mathbb{B}^n} \frac{\beta_{\mathbf{x},\mathbf{z}} - \beta^*_{\mathbf{x},\mathbf{z}}}{2} \hat{W}(\mathbf{x},\mathbf{z})  \right) \\
    & = \frac{1}{2^n} X \otimes \left(\sum_{\mathbf{x},\mathbf{z}\in\mathbb{B}^n} \Re{\beta_{\mathbf{x},\mathbf{z}}}\hat{W}(\mathbf{x},\mathbf{z}) \right) - 
    \frac{1}{2^n} \imath XZ  \otimes \left(\sum_{\mathbf{x},\mathbf{z}\in\mathbb{B}^n} \Im{\beta_{\mathbf{x},\mathbf{z}}}\hat{W}(\mathbf{x},\mathbf{z})\right) \\
    & = \frac{1}{2^n} \sum_{\mathbf{x},\mathbf{z}\in\mathbb{B}^n} \left( \Re{\beta_{\mathbf{x},\mathbf{z}}} X - \imath\Im{\beta_{\mathbf{x},\mathbf{z}}} XZ \right) \otimes \hat{W}(\mathbf{x},\mathbf{z}) =\\
    & = \frac{1}{2^n} \sum_{\mathbf{x},\mathbf{z}\in\mathbb{B}^n} \left( \Re{\beta_{\mathbf{x},\mathbf{z}}} X - \Im{\beta_{\mathbf{x},\mathbf{z}}} Y \right) \otimes \hat{W}(\mathbf{x},\mathbf{z}).
\end{split}
\label{proof_H_terms}
\end{equation}

Since matrix $B$ is real valued then according to \eqref{eq: real_condition} 
\begin{equation}
    \begin{cases}
        \Im{\beta_{\mathbf{x},\mathbf{z}}} = 0  , \text{if } \mathbf{x}\cdot\mathbf{z} = 0 \Mod{2}, \\
        \Re{\beta_{\mathbf{x},\mathbf{z}}} = 0,  \text{if } \mathbf{x}\cdot\mathbf{z} = 1 \Mod{2}.
    \end{cases}
\end{equation}

Now in the equation \eqref{proof_H_terms} since either $\Re{\beta_{\mathbf{x},\mathbf{z}}}=0$ or $\Im{\beta_{\mathbf{x},\mathbf{z}}}=0$, for $\mathbf{x},\mathbf{z}\in\mathbb{B}^n$ the number of terms in the decomposition of $H$ is the same as that for $B$. 
From \eqref{proof_H_terms} and the decomposition in proposition \ref{theor: sets} it follows that
\begin{align*} \label{eq:H_decomposition}
    S_0 &= \{X\} \bigotimes_1^{n} \{I,Z\} &(m=0)\\
    S_m &= \{\widehat{X,Y}\} \bigotimes_1^{n-m} \{I,Z\} \bigotimes_1^{m} \{X,Y\}, \;\; &
    m=1, \dots, n-1\\
    S_n &= \{\widehat{X,Y}\} \bigotimes_1^{m} \{X,Y\}, &
    (m=n)
\end{align*}
i.e the commuting subsets are constructed by appending $X$ or $Y$ at the beginning of the Pauli strings in the decomposition to ensure the required parity.
\end{proof}

%% file: main.bbl
\providecommand{\noopsort}[1]{}\providecommand{\singleletter}[1]{#1}%
\begin{thebibliography}{28}%
\makeatletter
\providecommand \@ifxundefined [1]{%
 \@ifx{#1\undefined}
}%
\providecommand \@ifnum [1]{%
 \ifnum #1\expandafter \@firstoftwo
 \else \expandafter \@secondoftwo
 \fi
}%
\providecommand \@ifx [1]{%
 \ifx #1\expandafter \@firstoftwo
 \else \expandafter \@secondoftwo
 \fi
}%
\providecommand \natexlab [1]{#1}%
\providecommand \enquote  [1]{``#1''}%
\providecommand \bibnamefont  [1]{#1}%
\providecommand \bibfnamefont [1]{#1}%
\providecommand \citenamefont [1]{#1}%
\providecommand \href@noop [0]{\@secondoftwo}%
\providecommand \href [0]{\begingroup \@sanitize@url \@href}%
\providecommand \@href[1]{\@@startlink{#1}\@@href}%
\providecommand \@@href[1]{\endgroup#1\@@endlink}%
\providecommand \@sanitize@url [0]{\catcode `\\12\catcode `\$12\catcode `\&12\catcode `\#12\catcode `\^12\catcode `\_12\catcode `\%12\relax}%
\providecommand \@@startlink[1]{}%
\providecommand \@@endlink[0]{}%
\providecommand \url  [0]{\begingroup\@sanitize@url \@url }%
\providecommand \@url [1]{\endgroup\@href {#1}{\urlprefix }}%
\providecommand \urlprefix  [0]{URL }%
\providecommand \Eprint [0]{\href }%
\providecommand \doibase [0]{http://dx.doi.org/}%
\providecommand \selectlanguage [0]{\@gobble}%
\providecommand \bibinfo  [0]{\@secondoftwo}%
\providecommand \bibfield  [0]{\@secondoftwo}%
\providecommand \translation [1]{[#1]}%
\providecommand \BibitemOpen [0]{}%
\providecommand \bibitemStop [0]{}%
\providecommand \bibitemNoStop [0]{.\EOS\space}%
\providecommand \EOS [0]{\spacefactor3000\relax}%
\providecommand \BibitemShut  [1]{\csname bibitem#1\endcsname}%
\let\auto@bib@innerbib\@empty
\bibitem [{\citenamefont {Feynman}\ \emph {et~al.}(1982)\citenamefont {Feynman} \emph {et~al.}}]{feynman2018simulating}%
  \BibitemOpen
  \bibfield  {author} {\bibinfo {author} {\bibfnamefont {R.~P.}\ \bibnamefont {Feynman}} \emph {et~al.},\ }\href@noop {} {\bibfield  {journal} {\bibinfo  {journal} {Int. J. Theor. phys}\ }\textbf {\bibinfo {volume} {21}} (\bibinfo {year} {1982})}\BibitemShut {NoStop}%
\bibitem [{\citenamefont {Lloyd}(1996)}]{lloyd1996universal}%
  \BibitemOpen
  \bibfield  {author} {\bibinfo {author} {\bibfnamefont {S.}~\bibnamefont {Lloyd}},\ }\href@noop {} {\bibfield  {journal} {\bibinfo  {journal} {Science}\ }\textbf {\bibinfo {volume} {273}},\ \bibinfo {pages} {1073} (\bibinfo {year} {1996})}\BibitemShut {NoStop}%
\bibitem [{\citenamefont {Berry}\ \emph {et~al.}(2015{\natexlab{a}})\citenamefont {Berry}, \citenamefont {Childs},\ and\ \citenamefont {Kothari}}]{Berry}%
  \BibitemOpen
  \bibfield  {author} {\bibinfo {author} {\bibfnamefont {D.~W.}\ \bibnamefont {Berry}}, \bibinfo {author} {\bibfnamefont {A.~M.}\ \bibnamefont {Childs}}, \ and\ \bibinfo {author} {\bibfnamefont {R.}~\bibnamefont {Kothari}},\ }\href@noop {} {\bibfield  {journal} {\bibinfo  {journal} {2015 IEEE 56th annual symposium on foundations of computer science}\ ,\ \bibinfo {pages} {792809}} (\bibinfo {year} {2015}{\natexlab{a}})}\BibitemShut {NoStop}%
\bibitem [{\citenamefont {Low}\ and\ \citenamefont {Chuang}(2019)}]{Low2019hamiltonian}%
  \BibitemOpen
  \bibfield  {author} {\bibinfo {author} {\bibfnamefont {G.~H.}\ \bibnamefont {Low}}\ and\ \bibinfo {author} {\bibfnamefont {I.~L.}\ \bibnamefont {Chuang}},\ }\href {\doibase 10.22331/q-2019-07-12-163} {\bibfield  {journal} {\bibinfo  {journal} {{Quantum}}\ }\textbf {\bibinfo {volume} {3}},\ \bibinfo {pages} {163} (\bibinfo {year} {2019})}\BibitemShut {NoStop}%
\bibitem [{\citenamefont {Childs}\ \emph {et~al.}(2003)\citenamefont {Childs}, \citenamefont {Cleve}, \citenamefont {Deotto}, \citenamefont {Farhi}, \citenamefont {Gutmann},\ and\ \citenamefont {Spielman}}]{childs2003exponential}%
  \BibitemOpen
  \bibfield  {author} {\bibinfo {author} {\bibfnamefont {A.~M.}\ \bibnamefont {Childs}}, \bibinfo {author} {\bibfnamefont {R.}~\bibnamefont {Cleve}}, \bibinfo {author} {\bibfnamefont {E.}~\bibnamefont {Deotto}}, \bibinfo {author} {\bibfnamefont {E.}~\bibnamefont {Farhi}}, \bibinfo {author} {\bibfnamefont {S.}~\bibnamefont {Gutmann}}, \ and\ \bibinfo {author} {\bibfnamefont {D.~A.}\ \bibnamefont {Spielman}},\ }\href@noop {} {\bibfield  {journal} {\bibinfo  {journal} {Proceedings of the thirty-fifth annual ACM symposium on Theory of computing}\ ,\ \bibinfo {pages} {059068}} (\bibinfo {year} {2003})}\BibitemShut {NoStop}%
\bibitem [{\citenamefont {Berry}\ \emph {et~al.}(2015{\natexlab{b}})\citenamefont {Berry}, \citenamefont {Childs}, \citenamefont {Cleve}, \citenamefont {Kothari},\ and\ \citenamefont {Somma}}]{berry2015simulating}%
  \BibitemOpen
  \bibfield  {author} {\bibinfo {author} {\bibfnamefont {D.~W.}\ \bibnamefont {Berry}}, \bibinfo {author} {\bibfnamefont {A.~M.}\ \bibnamefont {Childs}}, \bibinfo {author} {\bibfnamefont {R.}~\bibnamefont {Cleve}}, \bibinfo {author} {\bibfnamefont {R.}~\bibnamefont {Kothari}}, \ and\ \bibinfo {author} {\bibfnamefont {R.~D.}\ \bibnamefont {Somma}},\ }\href@noop {} {\bibfield  {journal} {\bibinfo  {journal} {Physical review letters}\ }\textbf {\bibinfo {volume} {114}},\ \bibinfo {pages} {090502} (\bibinfo {year} {2015}{\natexlab{b}})}\BibitemShut {NoStop}%
\bibitem [{\citenamefont {Berry}\ and\ \citenamefont {Novo}(2016)}]{berry2016corrected}%
  \BibitemOpen
  \bibfield  {author} {\bibinfo {author} {\bibfnamefont {D.~W.}\ \bibnamefont {Berry}}\ and\ \bibinfo {author} {\bibfnamefont {L.}~\bibnamefont {Novo}},\ }\href@noop {} {\bibfield  {journal} {\bibinfo  {journal} {arXiv preprint arXiv:1606.03443}\ } (\bibinfo {year} {2016})}\BibitemShut {NoStop}%
\bibitem [{\citenamefont {Low}\ and\ \citenamefont {Chuang}(2017)}]{low2017optimal}%
  \BibitemOpen
  \bibfield  {author} {\bibinfo {author} {\bibfnamefont {G.~H.}\ \bibnamefont {Low}}\ and\ \bibinfo {author} {\bibfnamefont {I.~L.}\ \bibnamefont {Chuang}},\ }\href@noop {} {\bibfield  {journal} {\bibinfo  {journal} {Physical review letters}\ }\textbf {\bibinfo {volume} {118}},\ \bibinfo {pages} {010501} (\bibinfo {year} {2017})}\BibitemShut {NoStop}%
\bibitem [{\citenamefont {Suau}\ \emph {et~al.}(2021)\citenamefont {Suau}, \citenamefont {Staffelbach},\ and\ \citenamefont {Calandra}}]{suau2021practical}%
  \BibitemOpen
  \bibfield  {author} {\bibinfo {author} {\bibfnamefont {A.}~\bibnamefont {Suau}}, \bibinfo {author} {\bibfnamefont {G.}~\bibnamefont {Staffelbach}}, \ and\ \bibinfo {author} {\bibfnamefont {H.}~\bibnamefont {Calandra}},\ }\href@noop {} {\bibfield  {journal} {\bibinfo  {journal} {ACM Transactions on Quantum Computing}\ }\textbf {\bibinfo {volume} {2}},\ \bibinfo {pages} {1} (\bibinfo {year} {2021})}\BibitemShut {NoStop}%
\bibitem [{\citenamefont {Trotter}(1959)}]{Trotter}%
  \BibitemOpen
  \bibfield  {author} {\bibinfo {author} {\bibfnamefont {H.~F.}\ \bibnamefont {Trotter}},\ }\href@noop {} {\bibfield  {journal} {\bibinfo  {journal} {Proceedings of the American Mathematical Society}\ }\textbf {\bibinfo {volume} {10}},\ \bibinfo {pages} {545} (\bibinfo {year} {1959})}\BibitemShut {NoStop}%
\bibitem [{\citenamefont {Childs}\ \emph {et~al.}(2021)\citenamefont {Childs}, \citenamefont {Su}, \citenamefont {Tran}, \citenamefont {Wiebe},\ and\ \citenamefont {Zhu}}]{Childs_Trotter}%
  \BibitemOpen
  \bibfield  {author} {\bibinfo {author} {\bibfnamefont {A.~M.}\ \bibnamefont {Childs}}, \bibinfo {author} {\bibfnamefont {Y.}~\bibnamefont {Su}}, \bibinfo {author} {\bibfnamefont {M.~C.}\ \bibnamefont {Tran}}, \bibinfo {author} {\bibfnamefont {N.}~\bibnamefont {Wiebe}}, \ and\ \bibinfo {author} {\bibfnamefont {S.}~\bibnamefont {Zhu}},\ }\href@noop {} {\bibfield  {journal} {\bibinfo  {journal} {Physical Review X}\ }\textbf {\bibinfo {volume} {11}},\ \bibinfo {pages} {011020} (\bibinfo {year} {2021})}\BibitemShut {NoStop}%
\bibitem [{\citenamefont {Kawase}\ and\ \citenamefont {Fujii}(2023)}]{Kawase}%
  \BibitemOpen
  \bibfield  {author} {\bibinfo {author} {\bibfnamefont {Y.}~\bibnamefont {Kawase}}\ and\ \bibinfo {author} {\bibfnamefont {K.}~\bibnamefont {Fujii}},\ }\href@noop {} {\bibfield  {journal} {\bibinfo  {journal} {Computer Physics Communications}\ }\textbf {\bibinfo {volume} {288}},\ \bibinfo {pages} {108720} (\bibinfo {year} {2023})}\BibitemShut {NoStop}%
\bibitem [{\citenamefont {Van Den~Berg}\ and\ \citenamefont {Temme}(2020)}]{Berg}%
  \BibitemOpen
  \bibfield  {author} {\bibinfo {author} {\bibfnamefont {E.}~\bibnamefont {Van Den~Berg}}\ and\ \bibinfo {author} {\bibfnamefont {K.}~\bibnamefont {Temme}},\ }\href@noop {} {\bibfield  {journal} {\bibinfo  {journal} {Quantum}\ }\textbf {\bibinfo {volume} {4}},\ \bibinfo {pages} {322} (\bibinfo {year} {2020})}\BibitemShut {NoStop}%
\bibitem [{\citenamefont {Mukhopadhyay}\ \emph {et~al.}(2023)\citenamefont {Mukhopadhyay}, \citenamefont {Wiebe},\ and\ \citenamefont {Zhang}}]{mukhopadhyay2023synthesizing}%
  \BibitemOpen
  \bibfield  {author} {\bibinfo {author} {\bibfnamefont {P.}~\bibnamefont {Mukhopadhyay}}, \bibinfo {author} {\bibfnamefont {N.}~\bibnamefont {Wiebe}}, \ and\ \bibinfo {author} {\bibfnamefont {H.~T.}\ \bibnamefont {Zhang}},\ }\href@noop {} {\bibfield  {journal} {\bibinfo  {journal} {npj Quantum Information}\ }\textbf {\bibinfo {volume} {9}},\ \bibinfo {pages} {31} (\bibinfo {year} {2023})}\BibitemShut {NoStop}%
\bibitem [{\citenamefont {Lawrence}\ \emph {et~al.}(2002)\citenamefont {Lawrence}, \citenamefont {Brukner},\ and\ \citenamefont {Zeilinger}}]{unbiased}%
  \BibitemOpen
  \bibfield  {author} {\bibinfo {author} {\bibfnamefont {J.}~\bibnamefont {Lawrence}}, \bibinfo {author} {\bibfnamefont {{\v{C}}.}~\bibnamefont {Brukner}}, \ and\ \bibinfo {author} {\bibfnamefont {A.}~\bibnamefont {Zeilinger}},\ }\href@noop {} {\bibfield  {journal} {\bibinfo  {journal} {Physical Review A}\ }\textbf {\bibinfo {volume} {65}},\ \bibinfo {pages} {032320} (\bibinfo {year} {2002})}\BibitemShut {NoStop}%
\bibitem [{\citenamefont {Yen}\ \emph {et~al.}(2023)\citenamefont {Yen}, \citenamefont {Ganeshram},\ and\ \citenamefont {Izmaylov}}]{Izmaylov}%
  \BibitemOpen
  \bibfield  {author} {\bibinfo {author} {\bibfnamefont {T.-C.}\ \bibnamefont {Yen}}, \bibinfo {author} {\bibfnamefont {A.}~\bibnamefont {Ganeshram}}, \ and\ \bibinfo {author} {\bibfnamefont {A.~F.}\ \bibnamefont {Izmaylov}},\ }\href@noop {} {\bibfield  {journal} {\bibinfo  {journal} {npj Quantum Information}\ }\textbf {\bibinfo {volume} {9}},\ \bibinfo {pages} {14} (\bibinfo {year} {2023})}\BibitemShut {NoStop}%
\bibitem [{\citenamefont {Crawford}\ \emph {et~al.}(2021)\citenamefont {Crawford}, \citenamefont {van Straaten}, \citenamefont {Wang}, \citenamefont {Parks}, \citenamefont {Campbell},\ and\ \citenamefont {Brierley}}]{Cliques}%
  \BibitemOpen
  \bibfield  {author} {\bibinfo {author} {\bibfnamefont {O.}~\bibnamefont {Crawford}}, \bibinfo {author} {\bibfnamefont {B.}~\bibnamefont {van Straaten}}, \bibinfo {author} {\bibfnamefont {D.}~\bibnamefont {Wang}}, \bibinfo {author} {\bibfnamefont {T.}~\bibnamefont {Parks}}, \bibinfo {author} {\bibfnamefont {E.}~\bibnamefont {Campbell}}, \ and\ \bibinfo {author} {\bibfnamefont {S.}~\bibnamefont {Brierley}},\ }\href@noop {} {\bibfield  {journal} {\bibinfo  {journal} {Quantum}\ }\textbf {\bibinfo {volume} {5}},\ \bibinfo {pages} {385} (\bibinfo {year} {2021})}\BibitemShut {NoStop}%
\bibitem [{\citenamefont {Pozrikidis}(2014)}]{band}%
  \BibitemOpen
  \bibfield  {author} {\bibinfo {author} {\bibfnamefont {C.}~\bibnamefont {Pozrikidis}},\ }\href@noop {} {\bibfield  {journal} {\bibinfo  {journal} {Oxford University Press}\ } (\bibinfo {year} {2014})}\BibitemShut {NoStop}%
\bibitem [{\citenamefont {LeVeque}(1998)}]{LeVeque}%
  \BibitemOpen
  \bibfield  {author} {\bibinfo {author} {\bibfnamefont {R.~J.}\ \bibnamefont {LeVeque}},\ }\href@noop {} {\bibfield  {journal} {\bibinfo  {journal} {Draft version for use in AMath}\ }\textbf {\bibinfo {volume} {585}},\ \bibinfo {pages} {112} (\bibinfo {year} {1998})}\BibitemShut {NoStop}%
\bibitem [{\citenamefont {Cerezo}\ \emph {et~al.}(2021)\citenamefont {Cerezo}, \citenamefont {Arrasmith}, \citenamefont {Babbush}, \citenamefont {Benjamin}, \citenamefont {Endo}, \citenamefont {Fujii}, \citenamefont {McClean}, \citenamefont {Mitarai}, \citenamefont {Yuan}, \citenamefont {Cincio} \emph {et~al.}}]{cerezo2021variational}%
  \BibitemOpen
  \bibfield  {author} {\bibinfo {author} {\bibfnamefont {M.}~\bibnamefont {Cerezo}}, \bibinfo {author} {\bibfnamefont {A.}~\bibnamefont {Arrasmith}}, \bibinfo {author} {\bibfnamefont {R.}~\bibnamefont {Babbush}}, \bibinfo {author} {\bibfnamefont {S.~C.}\ \bibnamefont {Benjamin}}, \bibinfo {author} {\bibfnamefont {S.}~\bibnamefont {Endo}}, \bibinfo {author} {\bibfnamefont {K.}~\bibnamefont {Fujii}}, \bibinfo {author} {\bibfnamefont {J.~R.}\ \bibnamefont {McClean}}, \bibinfo {author} {\bibfnamefont {K.}~\bibnamefont {Mitarai}}, \bibinfo {author} {\bibfnamefont {X.}~\bibnamefont {Yuan}}, \bibinfo {author} {\bibfnamefont {L.}~\bibnamefont {Cincio}},  \emph {et~al.},\ }\href@noop {} {\bibfield  {journal} {\bibinfo  {journal} {Nature Reviews Physics}\ }\textbf {\bibinfo {volume} {3}},\ \bibinfo {pages} {625} (\bibinfo {year} {2021})}\BibitemShut {NoStop}%
\bibitem [{\citenamefont {Peruzzo}\ \emph {et~al.}(2014)\citenamefont {Peruzzo}, \citenamefont {McClean}, \citenamefont {Shadbolt}, \citenamefont {Yung}, \citenamefont {Zhou}, \citenamefont {Love}, \citenamefont {Aspuru-Guzik},\ and\ \citenamefont {O’brien}}]{peruzzo2014variational}%
  \BibitemOpen
  \bibfield  {author} {\bibinfo {author} {\bibfnamefont {A.}~\bibnamefont {Peruzzo}}, \bibinfo {author} {\bibfnamefont {J.}~\bibnamefont {McClean}}, \bibinfo {author} {\bibfnamefont {P.}~\bibnamefont {Shadbolt}}, \bibinfo {author} {\bibfnamefont {M.-H.}\ \bibnamefont {Yung}}, \bibinfo {author} {\bibfnamefont {X.-Q.}\ \bibnamefont {Zhou}}, \bibinfo {author} {\bibfnamefont {P.~J.}\ \bibnamefont {Love}}, \bibinfo {author} {\bibfnamefont {A.}~\bibnamefont {Aspuru-Guzik}}, \ and\ \bibinfo {author} {\bibfnamefont {J.~L.}\ \bibnamefont {O’brien}},\ }\href@noop {} {\bibfield  {journal} {\bibinfo  {journal} {Nature communications}\ }\textbf {\bibinfo {volume} {5}},\ \bibinfo {pages} {4213} (\bibinfo {year} {2014})}\BibitemShut {NoStop}%
\bibitem [{\citenamefont {Welch}\ \emph {et~al.}(2014)\citenamefont {Welch}, \citenamefont {Greenbaum}, \citenamefont {Mostame},\ and\ \citenamefont {Aspuru-Guzik}}]{Welsh}%
  \BibitemOpen
  \bibfield  {author} {\bibinfo {author} {\bibfnamefont {J.}~\bibnamefont {Welch}}, \bibinfo {author} {\bibfnamefont {D.}~\bibnamefont {Greenbaum}}, \bibinfo {author} {\bibfnamefont {S.}~\bibnamefont {Mostame}}, \ and\ \bibinfo {author} {\bibfnamefont {A.}~\bibnamefont {Aspuru-Guzik}},\ }\href@noop {} {\bibfield  {journal} {\bibinfo  {journal} {New Journal of Physics}\ }\textbf {\bibinfo {volume} {16}},\ \bibinfo {pages} {033040} (\bibinfo {year} {2014})}\BibitemShut {NoStop}%
\bibitem [{\citenamefont {Berry}\ \emph {et~al.}(2007)\citenamefont {Berry}, \citenamefont {Ahokas}, \citenamefont {Cleve},\ and\ \citenamefont {Sanders}}]{berry2007efficient}%
  \BibitemOpen
  \bibfield  {author} {\bibinfo {author} {\bibfnamefont {D.~W.}\ \bibnamefont {Berry}}, \bibinfo {author} {\bibfnamefont {G.}~\bibnamefont {Ahokas}}, \bibinfo {author} {\bibfnamefont {R.}~\bibnamefont {Cleve}}, \ and\ \bibinfo {author} {\bibfnamefont {B.~C.}\ \bibnamefont {Sanders}},\ }\href@noop {} {\bibfield  {journal} {\bibinfo  {journal} {Communications in Mathematical Physics}\ }\textbf {\bibinfo {volume} {270}},\ \bibinfo {pages} {359} (\bibinfo {year} {2007})}\BibitemShut {NoStop}%
\bibitem [{\citenamefont {Tranter}\ \emph {et~al.}(2019)\citenamefont {Tranter}, \citenamefont {Love}, \citenamefont {Mintert}, \citenamefont {Wiebe},\ and\ \citenamefont {Coveney}}]{tranter2019ordering}%
  \BibitemOpen
  \bibfield  {author} {\bibinfo {author} {\bibfnamefont {A.}~\bibnamefont {Tranter}}, \bibinfo {author} {\bibfnamefont {P.~J.}\ \bibnamefont {Love}}, \bibinfo {author} {\bibfnamefont {F.}~\bibnamefont {Mintert}}, \bibinfo {author} {\bibfnamefont {N.}~\bibnamefont {Wiebe}}, \ and\ \bibinfo {author} {\bibfnamefont {P.~V.}\ \bibnamefont {Coveney}},\ }\href@noop {} {\bibfield  {journal} {\bibinfo  {journal} {Entropy}\ }\textbf {\bibinfo {volume} {21}},\ \bibinfo {pages} {1218} (\bibinfo {year} {2019})}\BibitemShut {NoStop}%
\bibitem [{\citenamefont {Costa}\ \emph {et~al.}(2019)\citenamefont {Costa}, \citenamefont {Jordan},\ and\ \citenamefont {Ostrander}}]{Costa}%
  \BibitemOpen
  \bibfield  {author} {\bibinfo {author} {\bibfnamefont {P.~C.~S.}\ \bibnamefont {Costa}}, \bibinfo {author} {\bibfnamefont {S.}~\bibnamefont {Jordan}}, \ and\ \bibinfo {author} {\bibfnamefont {A.}~\bibnamefont {Ostrander}},\ }\href {\doibase 10.1103/PhysRevA.99.012323} {\bibfield  {journal} {\bibinfo  {journal} {Phys. Rev. A}\ }\textbf {\bibinfo {volume} {99}},\ \bibinfo {pages} {012323} (\bibinfo {year} {2019})}\BibitemShut {NoStop}%
\bibitem [{\citenamefont {Yuster}\ and\ \citenamefont {Zwick}(2005)}]{yuster2005fast}%
  \BibitemOpen
  \bibfield  {author} {\bibinfo {author} {\bibfnamefont {R.}~\bibnamefont {Yuster}}\ and\ \bibinfo {author} {\bibfnamefont {U.}~\bibnamefont {Zwick}},\ }\href@noop {} {\bibfield  {journal} {\bibinfo  {journal} {ACM Transactions On Algorithms (TALG)}\ }\textbf {\bibinfo {volume} {1}},\ \bibinfo {pages} {2} (\bibinfo {year} {2005})}\BibitemShut {NoStop}%
\bibitem [{git()}]{github_page}%
  \BibitemOpen
  \href@noop {} {\enquote {\bibinfo {title} {Github page with numerical experiments},}\ }\bibinfo {howpublished} {\url{https://github.com/barseniev/tridiagonal-matrix-decomposition-quantum-simulation}}\BibitemShut {NoStop}%
\bibitem [{\citenamefont {Zacharov}\ \emph {et~al.}(2019)\citenamefont {Zacharov}, \citenamefont {Arslanov}, \citenamefont {Gunin}, \citenamefont {Stefonishin}, \citenamefont {Bykov}, \citenamefont {Pavlov}, \citenamefont {Panarin}, \citenamefont {Maliutin}, \citenamefont {Rykovanov},\ and\ \citenamefont {Fedorov}}]{zhores}%
  \BibitemOpen
  \bibfield  {author} {\bibinfo {author} {\bibfnamefont {I.}~\bibnamefont {Zacharov}}, \bibinfo {author} {\bibfnamefont {R.}~\bibnamefont {Arslanov}}, \bibinfo {author} {\bibfnamefont {M.}~\bibnamefont {Gunin}}, \bibinfo {author} {\bibfnamefont {D.}~\bibnamefont {Stefonishin}}, \bibinfo {author} {\bibfnamefont {A.}~\bibnamefont {Bykov}}, \bibinfo {author} {\bibfnamefont {S.}~\bibnamefont {Pavlov}}, \bibinfo {author} {\bibfnamefont {O.}~\bibnamefont {Panarin}}, \bibinfo {author} {\bibfnamefont {A.}~\bibnamefont {Maliutin}}, \bibinfo {author} {\bibfnamefont {S.}~\bibnamefont {Rykovanov}}, \ and\ \bibinfo {author} {\bibfnamefont {M.}~\bibnamefont {Fedorov}},\ }\href {\doibase https://doi.org/10.1515/eng—2019—0059} {\bibfield  {journal} {\bibinfo  {journal} {Open Engineering}\ }\textbf {\bibinfo {volume} {9}} (\bibinfo {year} {2019}),\ https://doi.org/10.1515/eng—2019—0059}\BibitemShut {NoStop}%
\end{thebibliography}%
